\documentclass[a4paper,11pt]{article}
\usepackage[T1]{fontenc}
\usepackage[utf8]{inputenc}

\usepackage{oldgerm}
\usepackage{color}
\usepackage{bbm}
\usepackage{amssymb}
\usepackage{epsfig}
\usepackage{amsthm}
\usepackage{amsmath, mathrsfs,psfrag}

\renewcommand{\mathbb}[1]{\mathbbm{#1}}

\newcommand{\Cl}{\mathbbm{C}}
\newcommand{\Rl}{\mathbb{R}}
\newcommand{\Nl}{\mathbb{N}}
\newcommand{\Zl}{\mathbb{Z}}
\definecolor{lightgray}{rgb}{0.8,0.8,0.8}

\newcommand{\Om}{\Omega}
\newcommand{\om}{\omega}

\newcommand{\te}{\theta}

\newcommand{\la}{\lambda}

\newcommand{\eps}{\varepsilon}

\newcommand{\A}{\mathcal{A}}

\newcommand{\B}{\mathcal{B}}

\newcommand{\M}{\mathcal{M}}
\newcommand{\I}{\mathcal{I}}
\newcommand{\K}{\mathcal{K}}

\newcommand{\NN}{\mathcal{N}}

\newcommand{\Hil}{\mathcal{H}}

\newcommand{\E}{\mathcal{E}}

\newcommand{\F}{\mathcal{F}}

\newcommand{\Ss}{\mathscr{S}}   % Schwartz space

\newcommand{\fti}{\tilde{f}}

\newcommand{\fhat}{\hat{f}}

\newcommand{\OO}{O}

\newcommand{\Sp}{\text{Sp}\,}
\DeclareMathOperator{\supp}{supp}

\newcommand{\dom}{\mathrm{dom}\,}

\newcommand{\Strip}{\mathrm{S}}
\DeclareMathOperator{\im}{Im}

\usepackage{amsmath}
\usepackage{amsfonts}
\usepackage{amssymb}
\usepackage{mathrsfs}
\usepackage[pdfborder={0 0 0}]{hyperref}
\usepackage{color}
\usepackage[T1]{fontenc}

\newtheorem{theorem}{Theorem}[section]
\newtheorem{proposition}[theorem]{Proposition}
\newtheorem{lemma}[theorem]{Lemma}
\newtheorem{corollary}[theorem]{Corollary}
\newtheorem{definition}[theorem]{Definition}

\newtheorem{example}[theorem]{Example}

\numberwithin{equation}{section}
\newlength{\dinwidth}
\newlength{\dinmargin}

\title{Localization in Nets of Standard Spaces}

\author{{\sc Gandalf Lechner}
\\
Institut f\"ur Theoretische Physik, Universit\"at Leipzig,
\\ Br\"uderstra\ss e 16, D-04103 Leipzig, Germany
\\
\phantom{X}\\
{\sc Roberto Longo}
\\
Dipartimento di Matematica,
Universit\`a di Roma ``Tor Vergata'',\\
Via della Ricerca Scientifica 1, I-00133 Roma, Italy
}

\date{March 05, 2014}

\begin{document}
\maketitle
%%================================================

%%================================================
\begin{abstract}
	Starting from a real standard subspace of a Hilbert space and a representation of the translation group with natural properties, we construct and analyze for each endomorphism of this pair a local, translationally covariant net of standard subspaces, on the lightray and on two-dimensional Minkowski space. These nets share many features with low-dimensional quantum field theory, described by corresponding nets of von Neumann algebras.
	
	Generalizing a result of Longo and Witten to two dimensions and massive multiplicity free representations, we characterize these endomorphisms in terms of specific analytic functions. Such a characterization then allows us to analyze the corresponding nets of standard spaces, and in particular compute their minimal localization length. The analogies and differences to the von Neumann algebraic situation are discussed.
\end{abstract}
%%================================================

%%%%%%%%%%%%%%%%%%%%%%%%%%%%%%%%%%%%
%% Contents
%%%%%%%%%%%%%%%%%%%%%%%%%%%%%%%%%%%%
\section{Introduction}

In relativistic quantum physics, the principle of locality requires spacelike observables to be commensurable and their associated operators to commute. In an operator-algebraic language, one can therefore view quantum field theories as certain collections (nets) of von Neumann algebras $\A(\OO)$, indexed by regions $\OO$ in spacetime, with the algebras $\A(\OO)$ satisfying specific inclusion, commutation and covariance properties. This formulation has traditionally been used in the investigation of model-independent aspects of quantum field theory \cite{Haag:1996,Araki:1999}.

However, in more recent years, a number of researchers have also begun to use this operator-algebraic framework for the (non-perturbative) construction of model theories, see \cite{SchroerWiesbrock:2000-1,BrunettiGuidoLongo:2002,BuchholzLechner:2004,KawahigashiLongo:2004,LongoRehren:2004,MundSchroerYngvason:2006,KawahigashiLongo:2006,Lechner:2008,LongoWitten:2010,Tanimoto:2011-1,Bischoff:2011,BuchholzLechnerSummers:2011,Lechner:2012,Plaschke:2012,Alazzawi:2012,LechnerSchlemmerTanimoto:2013,BarataJakelMund:2013} and references cited therein. Instead of the usual starting point of a classical Lagrangian density and techniques like quantization and renormalization, one works here in a quantum setting from the outset, and describes models in terms of certain algebraic data like endomorphisms or inclusions, or deformations thereof.

As some examples, we mention the operator-algebraic construction of free quantum field theories by modular localization \cite{BrunettiGuidoLongo:2002}, the construction of boundary quantum field theory models \cite{LongoRehren:2004}, the operator-algebraic formulation and solution of the inverse scattering problem for integrable quantum field theories \cite{SchroerWiesbrock:2000-1,Lechner:2003,Lechner:2008}, models of string-local infinite spin quantum fields \cite{MundSchroerYngvason:2006}, construction of local conformal nets by using framed vertex operator algebras \cite{KawahigashiLongo:2006}, models constructed from endomorphisms of standard pairs 
\cite{Tanimoto:2011-1,BischoffTanimoto:2011}, and deformations of quantum field theories \cite{BuchholzLechnerSummers:2011,Lechner:2012,Plaschke:2012,Alazzawi:2012,LechnerSchlemmerTanimoto:2013}.

In many of these approaches, one first considers a certain ``semi-local'' algebra $\M$, representing observables localized in a half space, or wedge, or half light ray. Such an algebra often gives rise to specific inclusions $\NN\subset\M$ or endomorphisms \cite{Borchers:1992,Wiesbrock:1992-2,ArakiZsido:2005,LongoWitten:2010,BuchholzLechnerSummers:2011}. In low dimensions, i.e. on a lightray or on two-dimensional Minkowski space, the step to strictly local objects associated with bounded regions in spacetime amounts to forming appropriate intersections or relative commutants $\NN'\cap\M$. Controlling the size of these intersections often represents the most difficult step in the construction program, and only very few general results are known in this context\footnote{With the notable exception of the special but important case of a split inclusion $\M_1\subset\M_2$ of two von Neumann algebras $\M_1,\M_2$ \cite{DoplicherLongo:1984}: Here spectral density conditions on modular operators (``modular nuclearity''
) \cite{BuchholzDAntoniLongo:1990-1} or Hamiltonians \cite{BuchholzWichmann:1986} give good criteria for non-triviality of the relative commutant $\M_1'\cap\M_2$, see \cite{BuchholzLechner:2004,Lechner:2008}.}.

In view of the difficulty of these questions, it can be helpful to consider them first in a simpler setting, where von Neumann algebras $\M\subset\B(\Hil)$ with a cyclic and separating vector $\Om$ are replaced by real subspaces $H:=\overline{\M_{\rm sa}\Om}\subset\Hil$, generated by their selfadjoint elements from a (cyclic and separating) vacuum vector $\Om$. In general, only partial information about the structure of $\M$ is encoded in $H$, but for interaction-free theories, the full von Neumann algebraic setting can be recovered from the subspace picture by second quantization \cite{Araki:1963,LeylandsRobertsTestard:1978}. The real subspace setting can be formulated in its own right, replacing commutants of von Neumann algebras with symplectic complements of real subspaces, and the main theorems of modular theory transfer to this setting \cite{Longo:2008}. It thus seems natural to also study questions of intersections or relative symplectic complements in this simplified setting first, as we 
shall do in the present article.

A natural description of a ``semi-local'' subspace is a standard pair $(H,T)$, consisting of a real subspace $H\subset\Hil$ as above (for precise definitions, see Section~\ref{Section:StandardPairs}) and a positive energy representation $T$ of the translations on $\Hil$, satisfying $T(x)H\subset H$ for $x$ in a half-ray or wedge. The endomorphisms of one-dimensional standard pairs $(H,T)$ were analyzed in \cite{LongoWitten:2010}. For irreducible standard pairs, the endomorphisms $V$ of $(H,T)$ were shown to be of the form $V=\varphi(P)$, where $P$ is the generator of $T$ and $\varphi$ a member of a specific family of analytic functions (symmetric inner functions on the upper half plane). This not only presents a clear characterization of the endomorphism semigroup $\E(H,T)$, but also an unexpected link to the construction of integrable quantum field theories from a scattering function, which satisfies almost identical properties as $\varphi$. Further exploring the link \cite{LechnerSchlemmerTanimoto:2013} 
between 
endomorphisms of standard pairs and integrable models, we will in this paper consider a 
construction of nets of standard subspaces, both in one and two dimensions. 

For the sake of a quick overview, we here describe our models in the one-dimensional (lightray) setting: We start from a standard pair $(H,T)$ as above, and  introduce for each endomorphism $V\in\E(H,T)$ a net $H_V$ taking intervals $I=(a,b)\subset\Rl$ to the real subspaces
\begin{align}\label{eq:IntroSubspaces}
	H_V(a,b)=T(b)H'\cap T(a)VH\,,
\end{align} 
where $H'$ denotes the symplectic complement of $H$ w.r.t. the imaginary part of the scalar product of $\Hil$. In Section \ref{Section:StandardPairs}, we discuss the basic properties of these nets in an abstract setting, and see in particular that their intersection properties are non-trivial. We then generalize the classification of the endomorphism semigroup obtained in \cite{LongoWitten:2010} to two dimensions, covering the irreducible and the massive multiplicity free case (Section \ref{Section:MultiplicityFree}). In the main Section~\ref{Section:1dCyclicity}, we consider the question for which intervals $I$ the space $H_V(I)$ is cyclic or at least non-trivial, in analogy to the intersection questions encountered in the von Neumann algebraic setting. Using techniques from complex analysis, in particular properties of inner functions and entire functions of exponential type, we show that for any endomorphism $V=\varphi(P)$ there exists a minimal localization radius $r_\varphi$ such that $H_V(I)$ is 
trivial if the 
length of $I$ is smaller than $2r_\varphi$, and cyclic if the length is larger than that. We consider the dependence $\varphi\mapsto r_\varphi$, and show how analytic properties of $\varphi$ (regular boundary behavior, density of zeros) influence the size of the subspaces \eqref{eq:IntroSubspaces}. In particular, we give examples of inner functions $\varphi$ realizing any value in $[0,\infty]$ as localization radius $r_\varphi$. This analysis relies on an explicit characterization of the ``undeformed'' local subspaces $H(I)$ in terms of entire functions, which is provided for the one- and two-dimensional case in the Appendix \ref{Section:Appendix}.

We discuss our results in the concluding Section~\ref{Section:Conclusions}. There we also compare with the von Neumann algebraic setting, the construction of integrable models, and known criteria for non-trivial intersections like the split property and modular nuclearity.

\section{Standard pairs and their endomorphism subnets}\label{Section:StandardPairs}

The main object of our investigation is a so-called {\em standard pair}. In preparation for its definition, recall that a closed real subspace $H$ of a complex Hilbert space $\Hil$ is called {\em cyclic} if $H+iH$ is dense in $\Hil$, and {\em separating} if $H\cap iH=\{0\}$. As usual, we denote by $H'$ the symplectic complement of $H$ w.r.t. the symplectic form $\im\langle\,\cdot\,,\,\cdot\,\rangle$, and recall that a closed real subspace $H$ is cyclic if and only if $H'$ is separating. A {\em standard subspace} $H\subset\Hil$ is a closed real subspace which is cyclic and separating. More information on standard subspaces, and proofs of the statements made here and further below in the text can be found in \cite{Longo:2008}.

We will also be concerned with unitary strongly continuous representations $T$ of $\Rl^d$ on $\Hil$, for the two cases $d=1$ and $d=2$. For $d=1$, we then have a self-adjoint generator $P$ such that $T(x)=e^{ixP}$, and for $d=2$, we introduce two generators $P_\pm$ of $T$ such that $T(x)=e^{i(x_+P_++x_-P_-)}$, where $x=(x_+,x_-)\in\Rl^2$ is presented in its light cone coordinates\footnote{Note that with these conventions, $x_\pm=x_0\pm x_1$ and $P_\pm=\frac{1}{2}(P_0\mp P_1)$, where $x_0,x_1$ and $P_0,P_1$ are the components of position and momentum w.r.t. the standard basis of two-dimensional Minkowski space.}. $T$ will be said to satisfy the {\em spectrum condition}, or to be a {\em positive energy representation}, if $P>0$ (for $d=1$), or if $P_+>0$ and $P_->0$ (for $d=2$). The two-dimensional positive energy representations thus consist of two commuting one-dimensional positive energy representations, which occasionally we shall denote by $T_\pm(a):=e^{iaP_\pm}$, $a\in\Rl$.

For concise formulations encompassing both the one- and two-dimensional case, it is also useful to introduce the sets, $d=1,2$,
\begin{align}
	W_d
	:=
	\begin{cases}
		\Rl_+ & d=1\\
		\{x\in\Rl^2\,:\,x_+>0,\,x_-<0\}
		& d=2
	\end{cases}
	\,.
\end{align}

\begin{definition}\label{Definition:StandardPair}
	Let $d\in\{1,2\}$.
	\begin{enumerate}
		\item A $d$-dimensional standard pair is a pair $(H,T)$ consisting of a standard subspace $H\subset\Hil$ of some complex Hilbert space $\Hil$ and a unitary strongly continuous representation $T$ of the translation group $\Rl^d$ on $\Hil$ such that
		\begin{itemize}
			\item $T(x)H\subset H$ for all $x\in W_d$.
			\item $T$ satisfies the spectrum condition.
		\end{itemize}
		\item A standard pair is called non-degenerate if there exists no non-zero vector in $\Hil$ which is invariant under all $T(x)$, $x\in\Rl^d$.
		\item An endomorphism of a standard pair is a unitary $V$ on $\Hil$ such that $VH\subset H$ and $[V,T(x)]=0$ for all $x\in\Rl^d$. The semigroup of all endomorphisms of $(H,T)$ is denoted $\E(H,T)$.
	\end{enumerate}
\end{definition}

In the literature \cite{LongoWitten:2010,LongoRehren:2011,BischoffTanimoto:2011}, the term ``standard pair'' is often reserved for what is here called a one-dimensional standard pair. In the present article, ``standard pair'' will refer to a pair in the sense of the above definition, with $d=1$ or $d=2$. As no confusion is likely to arise, we will use the same symbols $H,T,x$ in both the one- and two-dimensional situation. We also mention that a two-dimensional standard pair $(H,T)$ gives rise to a one-dimensional standard pair $(H,T_+)$ with the same standard space $H$, and representation $T_+$ instead of $T$. Similarly, the pair $(H',T_-)$ is standard w.r.t. $-W_1=\Rl_-$, i.e. satisfies the above definition with $W_1$ replaced by $-W_1$.

In the context of quantum field theory, standard pairs can be derived from so-called {\em Borchers triples} $(\M,T,\Om)$ \cite{Borchers:1992,BaumgrtelWollenberg:1992,BuchholzLechnerSummers:2011}, consisting of a von Neumann algebra $\M\subset\B(\Hil)$, a vector $\Om$ which is cyclic and separating for $\M$ (in the sense of operator algebras), and the same representation $T$ of translations as above, assumed to act by endomorphisms on $\M$ for $x\in W_d$. One then obtains a standard pair $(H,T)$ by defining $H$ as the closure of the real subspace $\{A\Om\,:\,A=A^*\in\M\}$.

Just as the modular Tomita Takesaki theory of von Neumann algebras is an essential tool in the context of Borchers triples, the modular theory of standard subspaces is an essential tool in the setting considered here. To any standard subspace $H$ one can associate a Tomita operator $S=S_H$, that is the anti-linear operator defined on the domain $D(S):=H+iH$ by
\begin{align}
	S : \xi+i\eta\mapsto \xi-i\eta\,,\qquad \xi,\eta\in H\,.
\end{align}
This map is well-defined and densely defined as a consequence of $H$ being standard, and clearly satisfies $S^2=1$ on $D(S)$. The operator $S$ is closed, and its adjoint is the Tomita operator of the symplectic complement, $S_H^*=S_{H'}$. As in the von Neumann algebra situation, one introduces the polar decomposition of $S$ as
\begin{align}
	S
	=
	J\Delta^{1/2}\,.
\end{align}
Here the {\em modular conjugation} $J$ is an anti-unitary involution, and the {\em modular operator} $\Delta=S^*S$ is a positive non-singular operator satisfying $J\Delta J=\Delta^{-1}$.

One then has an analogue of Tomita's theorem, 
\begin{align}\label{eq:TomitasTheorem}
	\Delta^{it}H=H\,,\quad t\in\Rl\,,\qquad J H =H'\,.
\end{align}

Also Borchers' theorem \cite{Borchers:1992} on the commutation relation of the modular group and the translations has an analogue for standard pairs \cite{Longo:2008}. In fact, $t\in\Rl$,
\begin{align}\label{eq:BorchersCommutationRelations}
	\Delta^{it}T(x)\Delta^{-it}&=T(\la_d(t)x)\,,\qquad
	JT(x)J=T(-x)\,,\qquad x\in\Rl^d\,,
\end{align}
where
\begin{align}
	\la_1(t)x&=e^{-2\pi t}x\,,\qquad\;\; x\in\Rl\,,\\
	(\la_2(t)x)_\pm&=e^{\mp 2\pi t}x_\pm\,,\qquad x=(x_+,x_-)\in\Rl^2\,.
\end{align}
We denote by $G_d$ the group generated by the translations $y\mapsto y+x$, $x,y\in\Rl^d$, and $\la_d(t)$, $t\in\Rl$. In the one-dimensional situation, $G_1$ is thus the affine group of $\Rl$, i.e. the ``$ax+b$ group'', and in the two-dimensional situation, $G_2$ is the the proper orthochronous Poincar\'e group in two dimensions. Thanks to the commutation relations \eqref{eq:BorchersCommutationRelations}, we have a unitary strongly continuous representation $U$ of $G_d$ on $\Hil$, called the representation associated with $(H,T)$. We call a standard pair irreducible iff $U$ is irreducible.

The logarithms $\log P$ and $\log\Delta$ of the generators of $U(G_1)$ satisfy canonical commutation relations if $\log P$ is defined, i.e. in the non-degenerate case. By von Neumann uniqueness, $U$ is therefore a multiple of the unique irreducible, unitary, strongly continuous, positive energy representation $\hat{U}$ of $G_1$. In view of \eqref{eq:BorchersCommutationRelations}, $\hat{U}$ fixes also the conjugation $J$ up to a sign, and thus $H=\ker(1-J\Delta^{1/2})$ is uniquely fixed by $\hat{U}$, i.e. there exists only one one-dimensional standard pair with irreducible associated representation, up to unitary equivalence. In comparison, $G_2$ has infinitely many inequivalent irreducible representations, classified by their mass.  In Section~\ref{Section:1dCyclicity}, we will recall convenient realizations of these data.
\\
\\
\indent Having recalled these facts, we proceed to considering families of standard subspaces. Just as a Borchers triple gives rise to a net of von Neumann algebras \cite{Borchers:1992}, a standard pair gives rise to a net of standard subspaces. To describe it, we first define its index set $\I_d$, $d=1,2$, to be the family of all subsets of $\Rl^d$ which are translates of $W_d$ or $-W_d$, or intersections thereof, i.e. of the form
\begin{align}
	I_{x,y}
	:=
	(-W_d+y)\cap(W_d+x)
	\,,\qquad
	y-x\in W_d\,.
\end{align}
The condition $y-x\in W_d$ is equivalent to this intersection being non-empty. For $d=1$, the family $\I_1$ thus consists of all half lines and intervals in $\Rl$, and for $d=2$, the family $\I_2$ consists of all wedges and double cones in $\Rl^2$. We also agree to write $I\perp\tilde{I}$ for $I,\tilde{I}\in\I_d$ if $I$ and $\tilde{I}$ are disjoint (in case $d=1$) or if $I$ and $\tilde{I}$ are spacelike separated (in case $d=2$). When working in a purely two-dimensional setting, we will also use the more common notation $\OO_{x,y}$ instead of $I_{x,y}$ to denote the double cones.

To each $I\in\I_d$, we now define a real subspace $H(I)\subset\Hil$. For $x,y\in\Rl^d$, we set
\begin{align}\label{eq:Net}
	H(W_d+x)&:=T(x)H
	\,,\qquad
	H(-W_d+x):=T(x)H'
	\,,\\
	H(I_{x,y})&:=H(-W_d+y)\cap H(W_d+x)
	\,,\qquad y-x\in W_d\,.
	\label{eq:HI}
\end{align}

These maps mimic basic locality and covariance properties of quantum field theories, formulated as nets of von Neumann algebras with specific properties. In fact, we have the following.

\begin{proposition}\label{Proposition:BasicNetProperties}
	The map \eqref{eq:Net} satisfies, $I,\tilde{I}\in\I_d$,
	\begin{enumerate}
		\item $H(I)\subset\Hil$ is a closed real separating subspace,
		\item $H(I)\subset H(\tilde{I})$ if $I\subset\tilde{I}$,
		\item $H(I)\subset H(\tilde{I})'$ if $I\perp\tilde{I}$,
		\item $T(x)H(I)=H(I+x)$ for all $x\in\Rl^d$,
		\item $JH(I)=H(-I)$, $\Delta^{it}H(I)=H(\la_d(t)I)$ for all $t\in\Rl$.
	\end{enumerate}
\end{proposition}
All these statements are well-known straightforward consequences of the definitions, see for example \cite[Prop.~III.3]{Borchers:1992} for a proof in the von Neumann algebraic setting for $d=2$.

It also follows from the definition that the net built from a two-dimensional standard pair $(H,T)$ contains two one-dimensional nets, built from $(H,T_+)$ and $(H',T_-)$. More specifically,
\begin{align}\label{eq:2dcontains1d+}
	H_+(I^{(1)}_{a,b})
	:=
	T_+(a)H\cap T_+(b)H'
\end{align}
and
\begin{align}\label{eq:2dcontains1d-}
	H_-(I^{(1)}_{a,b})
	:=
	T_-(a)H'\cap T_-(b)H
\end{align}
are one-dimensional nets satisfying all properties of Proposition~\ref{Proposition:BasicNetProperties}, with $\la_1(t)$ replaced by $\la_1(-t)$ in {\em v)} for the second net. They are related to the net $I^{(2)}_{x,y}\mapsto H(I^{(2)}_{x,y})$ built from the two-dimensional standard pair $(H,T)$ via
\begin{align}\label{eq:2dNetContains1dNet+}
	H(I^{(2)}_{x,y})
	&\supset
	T_-(a)H_+(I^{(1)}_{x_+,y_+}),\qquad y_-\leq a\leq x_-\,,
	\\	
	H(I^{(2)}_{x,y})
	&\supset
	T_+(b)H_-(I^{(1)}_{y_-,x_-}),\qquad x_+\leq b\leq y_+\,.
	\label{eq:2dNetContains1dNet-}
\end{align}
These claims follow from the definitions of the nets $H,H_+,H_-$, making use of the fact that $T_+(a),T_-(-a)\in\E(H,T)$ for $a\geq0$. They express the geometric situation depicted in the following illustration.
\begin{center}
	\includegraphics[width=80mm]{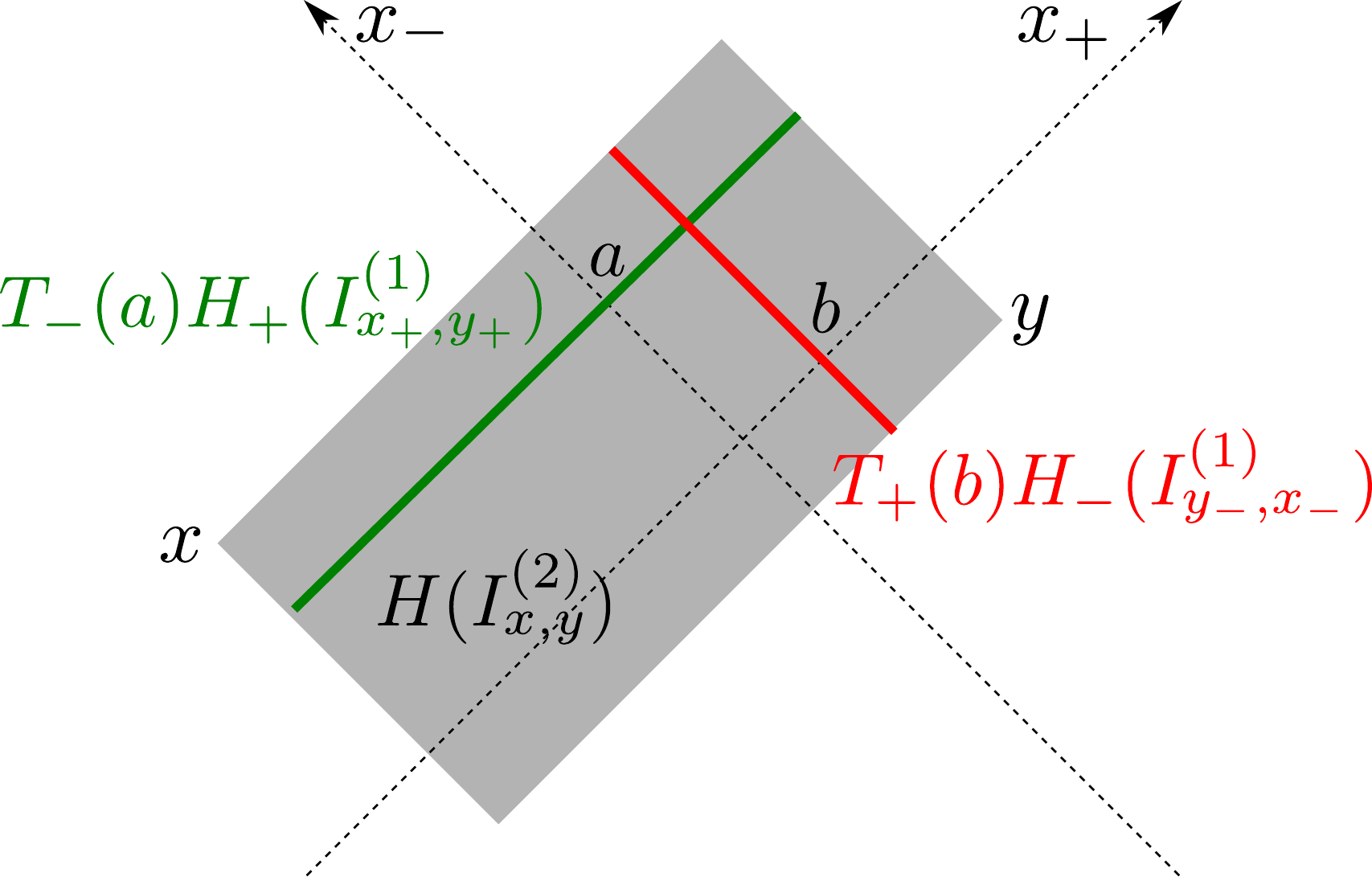}
\end{center}

Note that so far we have made no claim towards cyclicity of the local subspaces $H(I)$, $I\in\I_d$. This is done in the next proposition, which is closely related to \cite[Thm.~4.5]{BrunettiGuidoLongo:2002}.

\begin{proposition}\label{Proposition:CyclicityForHNet}
	Let $(H,T)$ be a non-degenerate standard pair. 
	\begin{enumerate}
		\item $H(I)$ is cyclic (and thus standard) for any $I\in\I_d$.
		\item $H$ is a factor, i.e.
			\begin{align}\label{eq:HIsAFactor}
				H\cap H'=\{0\}\,.
			\end{align}
	\end{enumerate}
\end{proposition}
\begin{proof}
	{\em i)} By translational covariance (Prop.~\ref{Proposition:BasicNetProperties}~{\em iv)}), it is sufficient to consider regions $I_{x,y}$ with $y=0$, $x\in -W_d$. In the one-dimensional irreducible situation, i.e. for the unique (non-degenerate) one-dimensional standard pair with associated irreducible representation $\hat{U}$ of $G_1$, cyclicity of $H(I_{x,0})$ is known to hold for all $x<0$ (exemplified by the model of the current algebra on the circle). 
	
	To analyze the general reducible one-dimensional case, we recall that the standard subspaces $H'$ and $T(x)H$ have the modular operators $S'=J\Delta^{-1/2}$ and $S_x=T(x)J\Delta^{1/2}T(-x)$, respectively, and that $H(I_{x,0})=H'\cap T(x)H$ is cyclic if and only if $\K_x:=\{\Psi\in\dom(S_xS')\,:\,S_xS'\Psi=\Psi\}$ is dense \cite[Prop.~4.1]{BrunettiGuidoLongo:2002}. Since $U$ is assumed to be non-degenerate, it decomposes into a direct sum $U=\bigoplus_k \hat{U}$. But $S_xS'\Psi=T(x)J\Delta^{1/2}T(-x)J\Delta^{-1/2}\Psi=T(x)\Delta^{-1/2}T(x)\Delta^{-1/2}\Psi$ depends only on the representation $U$. Thus also $\K_x$ decomposes into a direct sum of subspaces, each of which is dense. Thus $\K_x$ is dense and $H(I_{x,0})$ is cyclic for all $x<0$.
	
	\pagebreak
	In the two-dimensional case, we decompose the Hilbert space as $\Hil=\Hil_+\oplus\Hil_-$ with $\Hil_+:=\ker P_-$ and note that $U$ and $J$ decompose into direct sums. We now use that $G_2$ contains two copies of $G_1$, represented in $U$ by the modular group and the translations generated by $P_+$ or $P_-$. On  $\Hil_+$, the first representation of $G_1$ is non-degenerate because $U$ was assumed to be non-degenerate, i.e. $\ker P_+\cap\ker P_-=\{0\}$, and on $\Hil_-$, the second representation of $G_1$ is non-degenerate by construction. As also $H=H^+\oplus H^-$ decomposes into a direct sum, we have on both summands a one-dimensional net (generated by $(H^+,T_+|_{\Hil_+})$ and $(H^-,T_-|_{\Hil_-})$, respectively) with cyclic interval subspaces. But the double cone spaces of the net $H$ contain images of these interval subspaces under unitary transformations, see \eqref{eq:2dcontains1d+} and \eqref{eq:2dcontains1d-}. Thus also $H(I^{(2}_{x,y})$ is cyclic for any $y-x\in W_2$.
	
	{\em ii)} If $T$ is non-degenerate, then $U$ does not contain the trivial representation. It was shown in \cite[Thm.~2.5]{BrunettiGuidoLongo:2002} that in the two-dimensional case, this implies \eqref{eq:HIsAFactor}. But also in the one-dimensional case, we can consider the representation $U$ of $G_1$ as a representation of $G_2$, with one lightlike translation being represented trivially. This representation again does not contain the trivial one, so the conclusion also holds in this case.
\end{proof}

We now come to the main topic of this article, namely the construction of certain ``deformed'' versions of the one- or two-dimensional nets $\I_d\ni I\mapsto H(I)$. Similar to the generalization of \cite{LongoRehren:2004} obtained in \cite{LongoWitten:2010}, the input into this construction is an endomorphism $V\in\E(H,T)$. We define, $x,y\in\Rl^d$, 
\begin{align}\label{eq:NetV}
	H_V(W_d+x)&:=T(x)VH
	\,,\qquad
	H_V(-W_d+y):=T(y)H'
	\,,\\
	H_V(I_{x,y})
	&:=
	H_V(-W_d+y)\cap H_V(W_d+x)
	\,,\qquad y-x\in W_d\,.
\end{align}
For the trivial endomorphism $V=1$, this definition clearly coincides with the original construction \eqref{eq:Net}. For non-trivial $V$, we obtain proper subnets.

\begin{lemma}\label{Lemma:BasicNetPropertiesForHV}
	Let $(H,T)$ be a standard pair (one- or two-dimensional) and $V\in\E(H,T)$ an endomorphism.
	\begin{enumerate}
		\item $H_V(I)=H(I)\cap VH(I)\subset H(I)$ for any $I\in\I_d$.
		\item The properties i)--iv) of Proposition~\ref{Proposition:BasicNetProperties} hold if $H$ is replaced everywhere by $H_V$. Thus $I\mapsto H_V(I)$ is a translationally covariant local subnet of $I\mapsto H(I)$.
	\end{enumerate}
\end{lemma}
\begin{proof}
	{\em i)} We first observe that since $VH\subset H$ and $V$ is unitary, we have $H'\subset(VH)'=VH'$. Using this and the fact that $V$ commutes with the translations, we obtain for $x,y\in\Rl^d$, $y-x\in W_d$,
	 \begin{align*}
	 	H(I_{x,y})\cap VH(I_{x,y})
	 	&=
	 	\left( T(x)H\cap T(y)H' \right) \cap \left( T(x)VH\cap T(y)VH'\right)
	 	\\
	 	&=
	 	T(x)\left(H\cap VH \right) \cap T(y)\left(H'\cap VH'\right)
	 	\\
	 	&=
	 	T(x)VH\cap T(y)H'
	 	\\
	 	&=
	 	H_V(I_{x,y})\,.
	 \end{align*}
	{\em ii)} Properties {\em ii)--iv)} of Proposition~\ref{Proposition:BasicNetProperties} are straightforward to check with part {\em i)} of this lemma, the properties of the net $I\mapsto H(I)$, and the fact that $V$ commutes with $T$. The separating property {\em i)} follows since $H_V(I)\subset H(I)$.
\end{proof}

Whereas the properties corresponding to isotony, locality and translation covariance immediately covariance carry over from $H$ to its endomorphism subnets $H_V$, $V\in\E(H,T)$, the same is not true for the covariance under the modular data,  Proposition~\ref{Proposition:BasicNetProperties}~{\em v)}, and the cyclicity statement of Proposition~\ref{Proposition:CyclicityForHNet}. 

The $G_d$-covariance can be seen to be incompatible with $H_V$ being a proper subnet. The following lemma gathers the relevant statements.
\begin{lemma}
	Let $(H,T)$ be a standard pair and $V\in\E(H,T)$ an endomorphism such that $JVH=H'$ or $\Delta^{it}VH=VH$ for all $t\in\Rl$. Then $VH=H$.
\end{lemma}
\begin{proof}
	By \eqref{eq:TomitasTheorem}, $H'=JH$, i.e. $JVH=H'$ implies $VH=H$. Similarly, $\Delta^{it}VH=VH$ for all $t\in\Rl$ implies $VH=H$ because $VH\subset H$ \cite[Prop.~2.1.10]{Longo:2008}.
\end{proof}

The two conditions listed here, $JVH=H'$ or $\Delta^{it}VH=VH$, are particular cases of the covariance properties in Proposition~\ref{Proposition:BasicNetProperties}~{\em v)} (for $I=W_d$). Since $VH=H$ is equivalent to $H_V(I)=H(I)$ for all $I\in\I_d$ in view of the definition of the net $H_V$ \eqref{eq:NetV}, we see that $H_V$ will be $G_d$-covariant only in trivial cases.

We also give here a similar statement for later reference.

\begin{proposition}\label{unique}
	Let $(H,T)$ be a non-degenerate standard pair with irreducible associated representation $U$. If $K\subset\Hil$ is a proper real subspace of $\Hil$ such that  $T(t)K \subset K$ for all $t\in W_d$, and $\Delta_H^{is}K = K$, $s\in \mathbb R$, then $K = \alpha H$ for some $\alpha\in \mathbb C$, $|\alpha| = 1$.
\end{proposition}
\begin{proof}
	The closed complex linear span of $\cup_{t\in\mathbb R^d} T(t)K$ is an invariant subspace for the representation $U$ of $G_d$ associated with $(H,T)$; as $K\neq \{0\}$, it is equal to $\Hil$ by irreducibility. A Reeh-Schlieder type argument then shows that $K$ is cyclic. By applying this reasoning to $K'$ we see that $K$ is separating too.

	As $\Delta_H^{is}K = K$, we have that $\Delta_H^{is}$ and $\Delta_K^{it}$ commute; by the commutation relation \eqref{eq:BorchersCommutationRelations} $Z(s):=\Delta_H^{-is}\Delta_K^{is}$ is a one-parameter unitary group commuting with $T(x)$. Thus $Z$ commutes with $U$, so $Z(s) =e^{i\theta s}$ for some $\theta\in\mathbb R$ by the irreducibility of $U$. As $e^{i\theta s}K=Z(s)K = K$ we have $Z(s) =1$, namely $\Delta_K = \Delta_H$.

	By the commutation relations \eqref{eq:BorchersCommutationRelations}, we see that $J_K J_H$ is a unitary commuting with $U$, thus 
	$J_H J_K = z$ with $z\in\mathbb C$, $|z|=1$. Thus $S_K = zS_H$ and $K = \alpha H$ where $\alpha\in \mathbb C$, $\alpha^2 = z$.
\end{proof}

Whereas full $G_d$-covariance is ruled out in general, the situation is quite different for the cyclicity question of the local subspaces $H_V(I)$ -- these spaces might or might not be cyclic depending on $V$ and $I$. This is in close analogy to the situation encountered when generating nets of von Neumann algebras from a single algebra, where only under additional assumptions, like modular nuclearity \cite{BuchholzDAntoniLongo:1990-1} or the split property for wedges \cite{DoplicherLongo:1984}, cyclicity for intersections corresponding to double cones is known to hold \cite{BuchholzLechner:2004, Lechner:2008}. In fact, proving cyclicity for the intersections analogous to \eqref{eq:HI} is usually the hardest step in such a construction program.

It is the main aim of the present work to answer this question in the standard subspace setting, i.e. to find conditions on the endomorphism $V\in\E(H,T)$ guaranteeing the cyclicity of the interval/double cone subspaces $H(I_{x,y})$. We will comment in more detail on the relation to the von Neumann algebra setting in our conclusions.

To address the cyclicity question, it is helpful to have an explicit characterization of the endomorphism semigroup $\E(H,T)$ of a standard pair. For the one-dimensional irreducible case, such a characterization is known and will be recalled now. We begin with a general description of endomorphisms, ignoring the translations $T$. For its formulation, we use the notation $\Strip_{a}:=\{\zeta\in\Cl\,:\,0<\im\zeta<a\}$ for strips in the complex plane.

\pagebreak
\begin{lemma}{\bf \cite{ArakiZsido:2005, Longo:2008}}\label{Lemma:AnalyticityOfEndomorphisms}\\
	Let $H\subset\Hil$ be a standard subspace and $V\in\B(\Hil)$. The following are equivalent:
	\begin{enumerate}
		\item $VH\subset H$
		\item $JVJ\Delta^{1/2}\subset\Delta^{1/2}V$
		\item The map $\Rl\ni t\mapsto V(t):=\Delta^{-it}V\Delta^{it}$ extends to a bounded weakly continuous function on the closed strip $\overline{\Strip_{1/2}}$, analytic in the open strip $\Strip_{1/2}$, such that $V(\frac{i}{2})=JVJ$.
	\end{enumerate}
\end{lemma}

We will also need the concept of a symmetric inner function. 

\begin{definition}
	\begin{enumerate}
		\item A symmetric inner function on the strip $\Strip_\pi$ is a function $\psi:\Strip_\pi\to\Cl$ which is analytic and bounded and satisfies $\psi(\te+i\pi)=\overline{\psi(\te)}=\psi(\te)^{-1}$ for almost all $\te\in\Rl$.
		\item A symmetric inner function on the upper half plane $\Cl_+$ is a function $\varphi:\Cl_+\to\Cl$ which is analytic and bounded and satisfies $\varphi(-p)=\overline{\varphi(p)}=\varphi(p)^{-1}$ for almost all $p\geq0$.
	\end{enumerate}
\end{definition}

Two remarks are in order here: First, for bounded analytic functions $f$ on the strip $\Strip_\pi$ (respectively on the upper half plane), the limits $\lim_{\eps\searrow0}f(\te+i\eps)$ and $\lim_{\eps\searrow0}f(\te+i\pi-i\eps)$ (respectively just $\lim_{\eps\searrow0}f(q+i\eps)$) exist for almost all $\te\in\Rl$ and define boundary values in $L^\infty(\Rl)$, to which the statements in the above definition refer. As an inner symmetric function is uniquely determined by its boundary values and vice versa, we will often not distinguish between the analytic function and its boundary values in our notation. Second, it is clear that symmetric inner functions on the strip and on the half plane are in one to one correspondence by $\psi=\varphi\circ\exp$.

In several parts of our investigations we will encounter the transformation $\varphi\mapsto\gamma(\varphi)$, 
\begin{align}\label{eq:gamma}
	\gamma(\varphi)(p)
	:=
	\overline{\varphi(1/\overline{p})}\,,\qquad {\rm Im}\,p>0\,.
\end{align}
This transformation is seen to be an involutive automorphism of the semigroup of symmetric inner functions on the upper half plane. Those $\varphi$ which are invariant under $\gamma$ are called {\em $\gamma$-invariant}, or {\em scattering functions} because of their relation to inverse scattering theory \cite{Lechner:2003}, or -- for reasons that will become clear later -- {\em time-reflection invariant}.
\\
\\
In the one-dimensional irreducible case, any operator $V$ commuting with all translations $T(x)=e^{ixP}$ is a function $V=\varphi(P)$ of $P$. In view of Lemma~\ref{Lemma:AnalyticityOfEndomorphisms}, $\varphi$ has to have specific analytic properties. The precise statement is the following.

\begin{theorem}{\bf\cite{LongoWitten:2010}}\label{Theorem:EHTin1d}
	Let $(H,T)$ be a one-dimensional non-degenerate standard pair with irreducible associated representation $U$ of $G_1$. Then 
	\begin{align}
		\E(H,T)
		&=
		\{\varphi(P)\,:\,\varphi \text{ \rm symmetric inner on } \Cl_+\}\,,
	\end{align}
	where $P$ is the generator of $T$.
\end{theorem}

In the irreducible one-dimensional case, this theorem provides a characterization of the endomorphism semigroup $\E(H,T)$. As in the irreducible case also the modular conjugation $J$ is fixed up to a sign by $U=\hat{U}$, we can identify $H,J,\Delta,T$ with concrete function spaces and operators on them. This will be done in Section~\ref{Section:1dCyclicity}, where we investigate the cyclicity question in the one-dimensional irreducible case.

\section{Multiplicity free representations of the Poincar\'e group and endomorphisms of two-dimensional standard pairs}\label{Section:MultiplicityFree}

In this section, we generalize Theorem~\ref{Theorem:EHTin1d} to the two-dimensional case, both for irreducible and also certain reducible representations $U$ of $G_2$ associated with a two-dimensional standard pair $(H,T)$. 

We begin by recalling the structure of the irreducible, unitary, strongly continuous representations of $G_2$ which satisfy the spectrum condition. Requiring that $T=U|_{\Rl^2}$ is non-degenerate, such representations can be classified up to unitary equivalence according to the joint spectrum of the generators $P_\pm$, denoted $\Sp(U|_{\Rl^2})$. There are equivalence classes of three types:
\begin{enumerate}
	\item[$m$)] $\Sp(U|_{\Rl^2})=\{p\in\Rl^2\,:\,p_\pm>0,\;\;4\,p_+\, p_-=m^2\}$ for some $m>0$,
	\item[$0+$)] $\Sp(U|_{\Rl^2})=\{p\in\Rl^2\,:\,p_+\geq0,\,p_-=0\}$,
	\item[$0-$)] $\Sp(U|_{\Rl^2})=\{p\in\Rl^2\,:\,p_-\geq0,\,p_+=0\}$.
\end{enumerate}
For fixed positive mass $m>0$, we denote the unique representation of class $m$) by $U_m$. Similarly, $U_{0,\pm}$ denotes the unique representation of class $0\pm$).

Looking at the proof of Theorem~\ref{Theorem:EHTin1d} in \cite{LongoWitten:2010}, it becomes apparent that the essential property of the (there one-dimensional) standard pair used is that the translations generate a maximally abelian von Neumann algebra ${U|_{\Rl^2}}''$, so that the endomorphisms of $(H,T)$ must be functions of the translation generator. This property carries over to the two-dimensional irreducible situation without essential changes, and we have the same characterization of $\E(H,T)$ as in the one-dimensional case.

\begin{proposition}\label{proposition:MaximalAbelianInIrreps}
	\begin{enumerate}
		\item Let $U$ be an irreducible unitary strongly continuous positive energy representation of $G_2$, with $U|_{\Rl^2}$ non-degenerate. Then $U|_{\Rl^2}$ generates a maximally abelian von Neumann algebra ${U|_{\Rl^2}}''$.
		\item Let $(H,T)$ be a two-dimensional non-degenerate standard pair, with associated representation $U$ of $G_2$ irreducible. Then, if $U$ is of class $0\pm$),
		\begin{align}\label{eq:EndoSubgroup2DIrrep}
			\E(H,T)
			&=
			\{\varphi(\pm P_\pm)\,:\,\varphi \text{ \rm symmetric inner on } \Cl_+\}\,.
		\end{align}
		If $U$ is of class $m)$ with some $m>0$, this equation holds for both $P_+$ and $P_-$.
	\end{enumerate}
\end{proposition}
\begin{proof}
	{\em i)} If $U$ is of class $0\pm$), the same argument as in the one-dimensional case applies (see also \cite{Longo:2008}): The logarithm of the non-trivial momentum operator $P_\pm$ (which is non-singular because $U|_{\Rl^2}$ is non-degenerate) and the generator of the boosts form an irreducible representation of the canonical commutation relations. Hence the von Neumann algebras $\A_\pm$ generated by the translations $e^{iaP_\pm}$, $a\in\Rl$, are maximally abelian. As $U|_{\Rl^2}$ is abelian, this implies that also ${U|_{\Rl^2}}''$ is maximally abelian. 
	
	If $U$ is of class $m$) for some mass $m>0$, we can pick either generator $P_+$ or $P_-$ (since both are non-singular here) and apply the same argument.
	
	{\em ii)} For $U$ of class $0+$), this is the same proof as in \cite{LongoWitten:2010}, making use of part {\em i)} and Lemma~\ref{Lemma:AnalyticityOfEndomorphisms}~{\em iii)}. For class $0-$), the translations $e^{iaP_-}$ act by endomorphisms of $H$ for {\em negative} $a$, hence we get the same form as before after replacing $P_+$ by $-P_-$ in \eqref{eq:EndoSubgroup2DIrrep}. If $U$ is of class $m$) for some $m>0$, we can view an endomorphism $V\in\E(H,T)$ as a function of $P_+$ or $-P_-$ (both are non-singular), and thus both versions of \eqref{eq:EndoSubgroup2DIrrep} apply in this case.
\end{proof}

\noindent{\em Remark:} If $U$ is of class $m)$ for $m>0$, the two momentum operators $P_\pm$ are related by $P_-= \frac{1}{4}m^2\,P_+^{-1}$. Thus \eqref{eq:EndoSubgroup2DIrrep} gives an automorphism $\gamma_m$ of the symmetric inner functions on the upper half plane, $\gamma_m(\varphi)(p):=\varphi(-\frac{1}{4}m^2\,p^{-1})$, ${\rm Im}\,p>0$. In view of the symmetry $\varphi(-p)=\overline{\varphi(\overline{p})}$, this automorphism coincides with the involution $\gamma$ \eqref{eq:gamma}  up to the scaling factor $\frac{1}{4}m^2$.
\\
\\
The von Neumann algebra ${U|_{\Rl^2}}''$ can also be maximally abelian if $U$ is reducible, in particular in the massive case. (For convenience, we restrict ourselves to the massive case in the following. With minor modifications, our results also hold in the massless case.) We shall say that a unitary, positive energy representation $U$ of $G_2$ is {\em massive} if the boundary of the forward light cone has translation spectral measure zero. If $U$ is massive then it has an irreducible disintegration (unique as $G_2$ is of type $I$)
\begin{equation}\label{eq:UDirectIntegral}
	U
	=
	\int_{(0,\infty)}^\oplus N(m)U_{m}d\mu(m)\,.
\end{equation}
Here $\mu$ is a Borel measure on $\Rl_+$, and $N:\Rl_+\to\Nl\cup\{+\infty\}$ a measurable function giving the multiplicity $N(m)$ of $U_m$ in $U$. In this situation, we shall say that $U$ is \emph{multiplicity free} if $N(m)\in\{0,1\}$ for $\mu$-almost all $m$. We will also make use of the {\em mass operator} 
\begin{align}\label{eq:M}
	M:=2\,(P_+\cdot P_-)^{1/2}
\end{align}
associated with a representation $U$ of $G_2$.

\begin{lemma}\label{lemma:MultiplicityFreeVsMaximallyAbelian}
	Let $U$ be a massive representation of $G_2$.
	\begin{enumerate}
		\item The following are equivalent:
		\begin{enumerate}
			\item[a)] $U$ is multiplicity free.
			\item[b)] $U|_{\Rl^2}$ is multiplicity free.
			\item[c)] ${U|_{\Rl^2}}''$ is maximally abelian.
		\end{enumerate}
	\item The spectrum of the mass operator coincides with the support of the measure $\mu$ in \eqref{eq:UDirectIntegral}.
	\end{enumerate}
\end{lemma}
\begin{proof}
	$i)$ $a)\Rightarrow b)$ Assume that $U$ is multiplicity free. The unitary representations of $\Rl^2$ given by $U_m |_{\Rl^2}$ are mutually inequivalent for different masses (because $\Sp(U_m|_{\Rl^2})\neq\Sp(U_{m'}|_{\Rl^2})$ for $m\neq m'$), and decompose into direct integrals $U_m|_{\Rl^2}=\int^\oplus_{{\rm Sp}U_m|_{\Rl^2}}u_p\,d\mu_m(p)$ into the one-dimensional irreducible mutually inequivalent $\Rl^2$-representations $u_p(x)=e^{ipx}$, $p\in{\Sp}U_m|_{\Rl^2}$, with the usual Lorentz invariant measure $\mu_m$. The $u_p$ occur with uniform multiplicity $N(m)$ by Lorentz covariance, and since $U_m$ carries multiplicity one, we have $N(m)=1$, i.e. $U|_{\Rl^2}$ is multiplicity free.
	
	$b)\Rightarrow c)$ $U|_{\Rl^2}$ being multiplicity free is equivalent to its commutant ${U|_{\Rl^2}}'$ being abelian. But since $\Rl^2$ is abelian, this implies that ${U|_{\Rl^2}}'={U|_{\Rl^2}}''$ is maximally abelian.
	
	$c)\Rightarrow a)$ If ${U|_{\Rl^2}}''$ is maximally abelian, then $U''\supset {U|_{\Rl^2}}''$ has abelian commutant, i.e. $U$ is multiplicity free.
% 	{\em i)} Assume that $U$ is multiplicity free. The unitary representations of $\Rl^2$ given by $U_m |_{\Rl^2}$ are mutually inequivalent for different masses (because $\Sp(U_m|_{\Rl^2})\neq\Sp(U_{m'}|_{\Rl^2})$ for $m\neq m'$), so also $U|_{\Rl^2}$ is multiplicity free. If on the other hand $U|_{\Rl^2}$ is multiplicity free, then so is $U$, because its multiplicities $N(m)$ in \eqref{eq:UDirectIntegral} are the same as the multiplicities of $U|_{\Rl^2}$. Furthermore, since $U$, and thus $U|_{\Rl^2}$, disintegrate as in \eqref{eq:UDirectIntegral}, it follows from Proposition~\ref{proposition:MaximalAbelianInIrreps} that ${U|_{\Rl^2}}''$ is maximally abelian.
% 	
% 	Conversely, if $U$ is not multiplicity free, we find a subset $B\subset\Rl_+$ with $\mu(B)>0$ and $N(m)\geq2$ for all $m\in B$. 	
% 	Thus there exist operators $Y_m$ commuting with $T_m:=N(m){U_m}|_{\Rl^2}$, but not contained in the von Neumann algebra generated by $T_m$. Considering direct integrals of the form 
% 	\begin{align*}
% 		Y=\int_B^\oplus Y_m\,d\mu(m)
% 	\end{align*}
% 	then shows that ${U|_{\Rl^2}}''$ is not maximally abelian.
% 	
% 	 Finally, if $M$ is multiplicity free as a selfadjoint operator, it generates a maximally abelian von Neumann algebra $\M$ \cite[Thm.~VII.5]{ReedSimon:1972}. As $M$ is a function of $P_+,P_-$, we have $\M\subset {U|_{\Rl^2}}''$. But because ${U|_{\Rl^2}}''$ is abelian, this implies that ${U|_{\Rl^2}}''$ is maximally abelian, which is equivalent to $U$ being multiplicity free.
% 	

	{\em ii)} For the irreducible class $m)$ representation $U_m$, we clearly have $M=m\cdot 1$. Thus the mass operator of a general massive representation is $M=\int^\oplus_{(0,\infty)}m\,1_m\,d\mu(m)$, where $1_m$ is the identity on the the $N(m)$-fold direct sum of the representation space of $U_m$. This shows ${\rm spec}\, M=\supp\mu$.
\end{proof}

We thus see that in case $U$ is either irreducible or massive and multiplicity free, a unitary $V$ commuting with $U|_{\Rl^2}$ must be an element of ${U|_{\Rl^2}}''$, i.e. a function of the generators $P_+,P_-$. We have the following generalization of Theorem~\ref{Theorem:EHTin1d}.

\begin{theorem}\label{inn0} 
	Let $(H,T)$ be a non-degenerate two-dimensional standard pair such that its associated representation of $G_2$ is massive and multiplicity free. Then the following are equivalent:
	\begin{enumerate}
		\item $V\in\E(H,T)$,
		\item $V = \psi(P_+, M)$ for some $\psi\in L^{\infty}(\Rl_+\times \Rl_+)$ such that $p\mapsto\psi(p, m)$ is (the boundary value of) a symmetric inner function on the upper half-plane for almost every $m>0$.
	\end{enumerate}
\end{theorem}
\begin{proof}
	By assumption, $U$ is of the form
	\begin{align*}
		U= \int_{(0,\infty)}^\oplus U_m\, d\mu(m)
	\end{align*}
	for a Borel measure $\mu$ on $\Rl_+$. Since $V$ commutes with the translations, and ${U|_{\Rl^2}}''$ is maximally abelian by Lemma~\ref{lemma:MultiplicityFreeVsMaximallyAbelian}, $V$ disintegrates in the same manner, 
	\begin{align*}
		V =\int_{(0,\infty)}^\oplus V_m\, d\mu(m) \,.
	\end{align*}
	For fixed $m>0$, we have by Proposition~\ref{proposition:MaximalAbelianInIrreps} $V_m=\varphi_m(P_{+,m})$ with some symmetric inner function on the upper half plane, where $P_{+,m}$ is the generator for $+$lightlike translations in the representation $U_m$. Making use of these direct integral decompositions and \eqref{eq:BorchersCommutationRelations}, we 	can now evaluate $\langle\Psi_1,\Delta^{-it}V\Delta^{it}\Psi_2\rangle$ for $\Psi_k=\int^\oplus_{(0,\infty)}\Psi_{1,m}\,d\mu(m)$, $k=1,2$, as
	\begin{align*}
		\langle\Psi_1,\Delta^{-it}V\Delta^{it}\Psi_2\rangle
		=
		\int_0^\infty\langle\Psi_{1,m},\varphi_m(e^{2\pi t}P_{+,m})\Psi_{2,m}\rangle_m\,d\mu(m)\,.
	\end{align*}
	By Lemma~\ref{Lemma:AnalyticityOfEndomorphisms}~{\em iii)}, $VH\subset H$ is equivalent to this being an analytic function of $t$ in the strip $\Strip_{1/2}$, which in turn is equivalent to $t\mapsto\varphi_m(e^{2\pi t}p_+)$ being analytic in that strip for almost all $m>0$, and all $p_+>0$. This shows that $VH\subset H$ is equivalent to $\psi:\Rl_+\times\Rl_+\to\Cl$, $\psi(p_+,m):=\varphi_m(p_+)$ having the claimed properties.
	
	Taking into account that also the mass operator disintegrates, 
	\begin{align*}
		M=\int_{(0,\infty)}^\oplus m\,{\rm id}_m\, d\mu(m)\,,\qquad P_+=\int_{(0,\infty)}^\oplus {P_{+,m}}\, d\mu(m) \,,
	\end{align*}
	we also see $\psi(P_+,M)=V$.
\end{proof}

We next give two examples of reducible multiplicity free representations appearing in quantum field theory. The first example models the single particle Hilbert space of a theory with several particle species: Take masses $0<m_1<....<m_N<\infty$ and consider a standard pair with associated representation of $G_2$ as the direct sum $U=\bigoplus_{k=1}^N U_{m_k}$. Then $U$ is reducible, massive, and multiplicity free. As the spectrum of $M$ is $\{m_1,...,m_N\}$ in this case, we see that endomorphisms $V\in\E(H,T)$ are given by $N$-tuples of symmetric inner functions.

Another example arises from a typical two-particle Hilbert space. Starting from the irreducible representation $U_{m_0}$ on $\Hil_{m_0}$ of mass $m_0>0$, we consider the symmetric tensor product $\Hil_2:=\Hil_{m_0}\otimes_+\Hil_{m_0}\subset\Hil_{m_0}\otimes\Hil_{m_0}$, and define
\begin{align}
	U_{m_0}^2
	:=
	(U_{m_0}\otimes U_{m_0})|_{\Hil_2}
\end{align}
as the restriction of the tensor product to the symmetric subspace. Note that the generators $P_{\pm}$ of this representation are related to the generators $P_{+,m_0}$ and $P_{-,m_0}=\frac{1}{4}\,m_0^2\,P_{+,m_0}^{-1}$ of $U_{m_0}|_{\Rl^2}$ via
\begin{align}\label{eq:SumOfMomenta}
	P_\pm=P_{\pm,m_0}\otimes 1+1\otimes P_{\pm,m_0}
	\,.
\end{align}

\begin{proposition}
	$U_{m_0}^2$ is a massive multiplicity free representation of $G_2$. Indeed
	\begin{align}
		U_{m_0}^2
		&=
		\int_{[2m_0,\infty)}^\oplus U_{m}\;dm \,,
	\end{align}
	i.e. all positive energy representations of mass $m\geq 2m_0$ occur, each  with multiplicity one.
\end{proposition}
\begin{proof}
	We have $P_+=P_{+,m_0}\otimes 1+1\otimes P_{+,m_0}$ and $P_-=\frac{m_0^2}{4}(P_{+,m_0}^{-1}\otimes 1+1\otimes P_{+,m_0}^{-1})$ \eqref{eq:SumOfMomenta}. But spectral values $p,q,p',q'\geq0$ satisfying $p+q=p'+q'$ and $1/p+1/q=1/p'+1/q'$ are always related by either $p=p',q=q'$ or $p=q',q=p'$. As the latter possibility corresponds to flipped tensor factors ({\em cf.} \eqref{eq:SumOfMomenta}), which are identified in the symmetric tensor product, we see that $U_{m_0}^2|_{\Rl^2}$ is multiplicity free. By Lemma~\ref{lemma:MultiplicityFreeVsMaximallyAbelian}~$i)$, this implies that $U_{m_0}^2$ is multiplicity free. 
	
	The squared mass operator $M^2=4P_+P_-$ takes the spectral values $p+q$ and $\frac{m_0^2}{4}(1/p+1/q)$ to $m_0^2(2+p/q+q/p)$, which ranges over $[4m_0^2,\infty)$ as $p,q$ vary in $\Rl_+$. Thus ${\rm spec}\, M=[2m_0,\infty)$, which implies the claim by Lemma~\ref{lemma:MultiplicityFreeVsMaximallyAbelian}~$ii)$.
\end{proof}

In physics terminology, the above lemma expresses the fact that in two dimensions, energy-momentum conservation implies absence of momentum transfer in scattering processes of two indistinguishable particles.

The symmetric tensor square representation $U_{m_0}^2$ arise on the two-particle level of second quantization, i.e. from symmetric (real) tensor products $(H\otimes_{+,\Rl}H, T\otimes_+ T)$ of massive irreducible standard pairs, where the representation associated with $(H,T)$ is $U_{m_0}$. In this case, one can also reformulate the characterization of their endomorphism semigroup given in  Theorem~\ref{inn0} in terms of other variables than mass and lightlike momentum, for example in terms of sums and differences of two rapidities. We refrain from giving the details here.
\section{Cyclicity of local subspaces}\label{Section:1dCyclicity}

We now take up the task of characterizing the local subspaces of the endomorphism net $H_V$ \eqref{eq:NetV} in terms of $V$, in particular with regard to their cyclicity. We will mostly be concerned with irreducible (one- or two-dimensional) standard pairs and their associated endomorphism nets, and therefore first recall convenient representation spaces. 

\begin{enumerate}
	\item {\em ``Rapidity representation''} $\Hil=L^2(\Rl, d\te)$,
	\begin{align}\label{identify}
		(T(a)\psi)(\te) &= e^{iae^\te}\psi(\te),
		\qquad (Z\psi)(\te)=\psi(-\te)\,,\\
		(\Delta^{it}\psi)(\te)&= \psi(\te-2\pi t),
		\qquad
		(J\psi)(\te)=\overline{\psi(\te)}.
	\end{align}
	\item {\em ``Lightray representation'':} $\Hil=L^2(\Rl_+,dp/p)$,
	\begin{align}\label{eq:MomentumRepresentation}
		(T(a)\psi)(p) &= e^{iap}\psi(p),
		\qquad (Z\psi)(p)=\psi(\tfrac{1}{p})\,,\\
		(\Delta^{it}\psi)(p)&= \psi(e^{-2\pi t}\,p),
		\qquad
		(J\psi)(p)=\overline{\psi(p)}.
	\end{align}
	\item {\em ``Momentum representation'':} $\Hil=L^2(\Rl,\om_m(p_1)^{-1}dp_1)$, $\om_m(p_1)=(p_1^2+m^2)^{1/2}$, 
	\begin{align}
		(T(a)\psi)(p_1)
		&=
		e^{i\,\frac{a}{m}(\om_m(p_1)+p_1)}\psi(p_1)\,,\qquad (Z\psi)(p_1)=\psi(-p_1)\,,\\
		(\Delta^{it}\psi)(p_1)
		&=
		\psi(\cosh(2\pi t)p_1-\sinh(2\pi t)\om_m(p_1))\,,\qquad
		(J\psi)(p_1)=\overline{\psi(p_1)}\,.
	\end{align}
\end{enumerate}
Operators denoted by the same symbol are unitarily equivalent, as can be seen by introducing unitaries performing the change of variables $p_1=\frac{m}{2}(p-\frac{1}{p})=m\sinh\te$, $p=\frac{1}{m}(\om_m(p_1)+p_1)=e^\te$. We will always distinguish between the different pictures by denoting the variable $\te$, $p$, or $p_1$, respectively.

$J$ and $\Delta$ are the modular data of the unique irreducible one-dimensional standard pair $(H,T)$, and $T$ is the corresponding translation representation. To characterize $H$ (in the rapidity representation), recall the {\em Hardy space} \cite{Duren:1970},
\begin{align}
	{\mathbb H}^2(\Strip_\pi)
	:=
	\{&\psi :\Strip_\pi\to\Cl\,\text{ analytic },
	\quad\sup_{0<\la<\pi}\int_\Rl d\te\,|\psi(\te+i\la)|^2<\infty\}
	\,.
\end{align}
Any $\psi\in {\mathbb H}^2(\Strip_\pi)$ has boundary values (on $\Rl$) which lie in $L^2(\Rl,d\te)$. We will often denote these boundary values by the same symbol $\psi$, and also consider ${\mathbb H}^2(\Strip_\pi)$ as the subspace of all functions $L^2(\Rl,d\te)$ which are boundary values of functions in ${\mathbb H}^2(\Strip_\pi)$.

The following lemma (proven in Appendix~\ref{Section:Appendix}) specifies $H$ in the rapidity picture. Characterizations in the other pictures can be obtained by change of variables.

\begin{lemma}\label{Lemma:CharacterizeH}
In the rapidity representation, $H=\{\psi\in {\mathbb H}^2(\Strip_{\pi})\,:\,\overline{\psi(\te+i\pi)}=\psi(\te)\text{ a.e.}\}$.
\end{lemma}

$T(a)$ and $\Delta^{it}$ generate the unique irreducible positive energy representation $\hat{U}$ of $G_1$. The unitary involution $Z$ \eqref{eq:MomentumRepresentation} satisfies the commutation relations
\begin{align}
	Z\Delta^{it}Z&=\Delta^{-it}\,,\qquad ZJ=JZ\,,\qquad
	ZT(a)Z=:T'(a)=e^{iaP^{-1}}\,,
\end{align}
and one has $T'(a)\in\E(H,T)$ for $a\leq0$ \cite{LongoWitten:2010}. The first two equations imply $ZH=H'$.

The two-dimensional irreducible representations of $G_2$ can now be described as follows: In the mass zero case, $U_{0,\pm}$ is generated by $\Delta^{it}$ and the translations $T_{0,\pm}$,
\begin{align}
	T_{0,+}(x):=T(x_+)\,,\qquad T_{0,-}(x):=T'(x_-)\,.
\end{align}
The massive irreducible representations $U_m$, $m>0$, are generated by $\Delta^{it}$ and the translations
\begin{align}\label{eq:TmAndT}
	T_m(x)=T(\tfrac{mx_+}{2})T'(\tfrac{mx_-}{2})\,,
\end{align}
with generators $P_+=\frac{m}{2}P$, $P_-=\frac{m}{2}P^{-1}$.

In the two-dimensional situation, the unitary $Z$ implements spatial reflection (parity), $ZT(x_0,x_1)Z=T(x_0,-x_1)$. Similarly, the antiunitary involution 
\begin{align}
	\Gamma:=ZJ
\end{align}
implements time reflection, $\Gamma T(x_0,x_1)\Gamma=T(-x_0,x_1)$.
\\
\\
Since we are considering irreducible standard pairs, the endomorphisms $V\in\E(H,T)$ are given by symmetric inner functions, $\varphi$ (Theorem~\ref{Theorem:EHTin1d}, Proposition~\ref{proposition:MaximalAbelianInIrreps}). We will denote the net $H_V$ given by the endomorphism $V=\varphi(P)$ (in the one-dimensional case) by $I\mapsto H_\varphi(I)$, and the net given by the endomorphism $V=\varphi(\pm P_\pm)$ (in the two-dimensional case) with representation $U_{0,\pm}$ or $U_m$, $m>0$), by $\OO\mapsto H^{0,\pm}_\varphi(\OO)$ and $\OO\mapsto H^m_\varphi(\OO)$, respectively.

The massless two-dimensional nets can not be distinguished from one-dimensional nets (chiral components). In fact, we have 
\begin{align*}
	H^{0,+}_\varphi(\OO_{a,b})
	&=
	T_{0,+}(a)\varphi(P_+)H\cap T_{0,+}(b)H'\\
	&=
	T(a_+)\varphi(P)H\cap T(b_+)H'\\
	&=
	H_{\varphi}(I_{a_+,b_+})
\end{align*}
and
\begin{align*}
	\Gamma H^{0,-}_\varphi(\OO_{a,b})
	&=
	JT(a_-)Z\varphi(-P^{-1})H\cap JT(b_-)ZH'\\
	&=
	T(-a_-)J\varphi(-P)H'\cap T(-b_-)H'\\
	&=
	T(-a_-)\varphi(P)H\cap T(-b_-)H'\\
	&=
	H_{\varphi}(I_{-a_-,-b_-})
	\,,
\end{align*}
and thus the nets $H^{0,\pm}_\varphi$ either coincide with $H_\varphi$ or its time reflection $\Gamma H_\varphi$. We will therefore work with the one-dimensional nets $H_\varphi$ and the massive two-dimensional nets $H^m_\varphi$ only.

Since we are interested in the size of local subspaces, we introduce the {\em minimal localization radii} $r_\varphi,r_{m,\varphi}$ as

\begin{definition}
	The minimal localization radii $r_\varphi$ and $r_{m,\varphi}$ are defined as
	\begin{align}
		r_\varphi
		&:=
		\inf\{r\geq0\,:\,H_\varphi(I_r)\neq\{0\}\}\,,\\
		r_{m,\varphi}
		&:=
		\inf\{r\geq0\,:\,H^m_\varphi(\OO_r)\neq\{0\}\}\,.
	\end{align}
\end{definition}

Both these numbers lie in $[0,\infty]$ (with infinity meaning trivial subspaces for intervals/double cones of any size), and we will later give examples of functions $\varphi$ realizing any value in $[0,\infty]$ as minimal localization radius $r_\varphi$ or $r_{m,\varphi}$. 

It is interesting to note that the minimal localization radius always marks a sharp divide between trivial and cyclic subspaces:

\begin{theorem}\label{Theorem:CyclicAboveR}
	\begin{enumerate}
		\item Let $r>r_\varphi$. Then $H_\varphi(I_r)$ is cyclic.
		\item Let $r>r_{m,\varphi}$. Then $H^m_\varphi(\OO_r)$ is cyclic.
	\end{enumerate}
\end{theorem}
The proof of this theorem is postponed to the subsequent subsections. We conclude the present general discussion by pointing out some relations between the one- and two-dimensional nets. Note that the relation between the time reflection $\Gamma$ and the automorphism $\gamma$ \eqref{eq:gamma} in part $ii)$ of the proposition below also explains the term ``time reflection invariant'' for the $\gamma$-invariant symmetric inner functions.

\begin{proposition}\label{Proposition:1dvs2dNets}
	\begin{enumerate}
		\item $H_\varphi^m(\OO_r)\supset H_{\varphi\circ\frac{m}{2}}(I_{\frac{mr}{2}})$
		\item $\Gamma H_\varphi^m(\OO_r)=H^m_{\gamma(\varphi)}(\OO_r)$
		\item $r_{m,\varphi}\leq\min\{r_\varphi,\,r_{\gamma(\varphi)}\}$.
	\end{enumerate}
\end{proposition}
\begin{proof}
	$i)$ Making use of \eqref{eq:TmAndT},  $P_+=\frac{m}{2}P$, and $T(a),T'(-a)\in\E(H,T)$ for $a\geq0$, we find the inclusion
	\begin{align*}
		H_\varphi^m(\OO_r)
		&=
		T(-\tfrac{mr}{2})T'(\tfrac{mr}{2})\varphi(P_+) H\cap T(\tfrac{mr}{2})T'(-\tfrac{mr}{2})H'
		\\
		&\supset
		T(-\tfrac{mr}{2})\varphi(\tfrac{m}{2}P)H\cap T(\tfrac{mr}{2})H'
		\\
		&=
		H_{\varphi\circ\frac{m}{2}}(I_{\frac{mr}{2}})\,,
	\end{align*}
	as claimed. $ii)$ follows in view of the identities $\Gamma T(a)=T'(-a)\Gamma$, $\Gamma H=H$, and $\Gamma\varphi(P)=\gamma(\varphi)(P)\Gamma$, by applying $\Gamma$ to the first line of the above calculation.
	
	$iii)$ By Borchers' commutation relations, we find for $t$ such that $e^{-2\pi t}\cdot\frac{m}{2}=1$
	\begin{align*}
		\Delta^{it}H^m_\varphi(\OO_r)
		&\supset
		\Delta^{it}H_{\varphi\circ\frac{m}{2}}(I_{\frac{mr}{2}})
		\\
		&=
		\Delta^{it}T(-\tfrac{mr}{2})\varphi(\tfrac{m}{2}P)H\cap \Delta^{it}T(\tfrac{mr}{2})H'
		\\
		&=
		T(-r)\varphi(P)H\cap T(r)H'
		\\
		&=
		H_\varphi(I_r)\,.
	\end{align*}
	As $\Delta^{it}H^m_\varphi(\OO_r)$ is non-trivial if and only if $H^m_\varphi(\OO_r)$ is, we find non-triviality for $r>r_\varphi$, i.e. have shown $r_{m,\varphi}\leq r_\varphi$. In view of part $ii)$, we also get $r_{m,\varphi}\leq r_{\gamma(\varphi)}$.
\end{proof}

We will see later that in certain cases, the inequality in $iii)$ becomes an equality, whereas in other cases, one has $r_\varphi=r_{\gamma(\varphi)}=\infty$, but $r_{m,\varphi}=0$. For such examples, and also for the proof of Theorem~\ref{Theorem:CyclicAboveR} and the calculation of $r_\varphi$ and $r_{m,\varphi}$ from $\varphi$, more detailed information on the subspaces $H(I_r)$, $H^m(\OO_r)$ is needed. The following two subsections are devoted to studying these two cases.

\subsection{Cyclicity of interval subspaces}

In this subsection, we analyze the interval subspaces $H_\varphi(I)$. As a prerequisite for this, we first give a characterization of the localized standard subspaces $H(I)$ for the case $V=1$. In view of the translational invariance, it suffices to consider the symmetric intervals $I=I_r:=(-r,r)$, $r>0$. Working in the lightray representation, it is useful to introduce the skew-symmetric extension of functions $\psi\in L^2(\Rl_+,dp/p)$ to $\Rl$ as
\begin{align*}
	\psi^{\tt s}(p) :=
	\begin{cases}
		\overline{\psi(-p)} & p<0\\
		\psi(p) & p>0		
	\end{cases}
	\,.
\end{align*}
$\psi^{\tt s}$ defines a tempered distribution in $\Ss'(\Rl)$ (see Appendix~\ref{Section:Appendix1d}), so that we can consider its (inverse) Fourier transform in the sense of distributions.

\begin{proposition}\label{Proposition:CharacterizationOfHI}
	Let $(H,T)$ be the one-dimensional irreducible standard pair, and $\psi\in\Hil = L^2(\mathbb R_+ , dp/p)$. The following are equivalent:
	\begin{enumerate}
		\item $\psi\in H(I_r)$.
		\item The inverse Fourier transform of $\psi^{\tt s}$ has support in $I_r$.
		\item $\psi$ extends to an entire analytic function such that $\psi(-\bar p) = \overline{\psi(p)}$ and
		$|\psi(p)| \leq C |p| e^{r|\im p|}$ for all $p\in\mathbb C$ and some constant $C>0$.
% 		\item $\psi$ extends to an entire analytic function such that $\psi(-\bar p) = \overline{\psi(p)}$ and for some $N\in\Nl$
% 		$|\psi(p)| \leq C (1 + |p|)^N e^{r|\im p|}$, $p\in\mathbb C$, for some constant $C>0$.
		\item $\psi$ extends to an entire analytic function such that $\psi(-\bar p) = \overline{\psi(p)}$, and of exponential type at most $r$, namely for all $\eps>0$ there exists $C_\eps>0$ such that $|\psi(p)| \leq C_\varepsilon  e^{(r+\varepsilon)|p|}$,  $p\in\mathbb C$.
	\end{enumerate}
	For $\psi\in H(I)$, the function $\Rl_+\ni p\mapsto -i\psi(p)/p$ is the restriction to $\Rl_+$ of the Fourier transform of a real function in $L^2(\Rl,dx)$ with support in the closure of $I$, and
	\begin{align}\label{eq:HIFourier}
		\{\widehat{f'}|_{\Rl_+}\,:\,f\in C_{c,\Rl}^\infty(I)\}
		\subset 
		H(I)
	\end{align}
	is cyclic.
\end{proposition}
Here and in the following, a subscript ``$c$'' means compact support and a subscript ``$\Rl$'' real-valued functions. Proposition~\ref{Proposition:CharacterizationOfHI} is proven in Appendix~\ref{Section:Appendix1d}.
\\
\\
After these preparations, we turn to the proof of Theorem~\ref{Theorem:CyclicAboveR}~$i)$. We first give an auxiliary Lemma.

\begin{lemma}\label{Lemma:PhiAsRatio-1d}
	Let $\varphi$ be a symmetric inner function, and $r>0$. Then  $H_\varphi(I_r)\neq\{0\}$ if and only if there exist non-zero $\psi_1,\psi_2\in H(I_r)$ such that $\varphi=\psi_1/\psi_2$.
\end{lemma}
\begin{proof}
	By Lemma~\ref{Lemma:BasicNetPropertiesForHV}~$i)$, $H_\varphi(I_r)=H(I_r)\cap \varphi(P)H(I_r)$, i.e. $H_\varphi(I_r)\neq\{0\}$ is equivalent to the condition that there exist non-zero $\psi_1,\psi_2\in H(I_r)$ such that $\psi_1=\varphi\cdot \psi_2\Rightarrow\varphi=\psi_1/\psi_2$.
\end{proof}

	\noindent{\em Proof of Theorem~\ref{Theorem:CyclicAboveR}~$i)$}. Let $r>r_\varphi$, and pick $r'$ such that $r_\varphi<r'<r$. By definition of $r_\varphi$, we have $H_\varphi(I_{r'})\neq\{0\}$, and by the preceding Lemma, this implies $\varphi=\psi_1/\psi_2$ with certain non-zero $\psi_1,\psi_2\in H(I_{r'})$. Now let $f\in C_{c,\Rl}^\infty(I_\eps)$, with $\eps:=r-r'>0$. We will show $\psi_k\cdot\widehat{f'}\in H(I_{r})$, $k=1,2$, by verifying Proposition~\ref{Proposition:CharacterizationOfHI}~$iv)$. To begin with, $\psi_k\cdot\widehat{f'}$ lies in $L^2(\Rl_+,dp/p)$ because $\psi_k$ does and $\widehat{f'}$ is of Schwartz class on $\Rl$. Now, by Proposition~\ref{Proposition:CharacterizationOfHI}, both $\psi_k$ and $\widehat{f'}$ extend to entire analytic functions, of exponential type at most $r'$ and $\eps$, respectively, and both satisfy the reality condition $\overline{\psi_k(-\overline{p})}=\psi_k(p)$. Hence their product has the same properties, with exponential type at most $r'+\eps=r$, i.e. $\psi_k\cdot\widehat{f'}\in H(I_{r})$ by Proposition~\ref{Proposition:CharacterizationOfHI}~$iv)$. This implies $\psi_1\widehat{f'}\in H_\varphi(I_r)$, since this vector also lies in $\varphi(P)H(I_r)$ because $\psi_1\widehat{f'}=\varphi\cdot\psi_2\widehat{f'}$.
	
	It remains to show that the $\psi_1\widehat{f'}$, where $f$ varies over $C_{c,\Rl}^\infty(I_\eps)$, span a cyclic space. Let $\eta$ be in the orthogonal complement of the complex linear span of these vectors, then $f_1,f_2\in C_{c,\Rl}^\infty(I_\eps)$,
	\begin{align*}
		0
		&=
		\langle\eta,\psi_1\cdot(\widehat{f_1'}+i\widehat{f_2'})\rangle
		=
		\langle\overline{\psi_1}\cdot \eta,(\widehat{f_1'}+i\widehat{f_2'})\rangle
		\,,\qquad
		f_1,f_2\in C^\infty_{c,\Rl}(I_\eps)\,.
	\end{align*}
	By Proposition~\ref{Proposition:CharacterizationOfHI}, the vectors in the right entry of the scalar product span a dense subspace of $\Hil$, i.e. $\overline{\psi_1}\cdot \eta=0$. But $\psi_1$ is entire and non-zero, so it has only isolated zeros on $\Rl$, which implies $\eta=0$.{\hfill $\square$ \\[2mm] \indent}

In the remainder of this subsection, we address the problem of determining the localization radius $r_\varphi$ from $\varphi$. For this, it will be essential to exploit the analytic structure of the elements of $H(I)$, as presented in Proposition~\ref{Proposition:CharacterizationOfHI}, and the analytic structure of $\varphi$. Recall that any symmetric inner function is of the form
\begin{align}\label{eq:PhiFactorization}
	\varphi(p)
	=
	\pm e^{ipx}\,B(p)\,S(p)\,,\qquad S(p):=\exp\left(-i\int_\Rl \frac{1+p\,t}{p-t}\,d\mu(t)\right)
	\,,
\end{align}
where $x\geq0$, $B$ is a symmetric Blaschke product for the upper half plane, and $\mu$ a Lebesgue singular finite symmetric measure on $\Rl$. In more detail, this means that $B$ is of the form
\begin{align}\label{eq:BlaschkeProduct}
	B(p)
	=
	\prod_n\frac{p-p_n}{p-\overline{p_n}}
	\,,
\end{align}
where the $p_n$ are the zeros of $p_n$. Since they occur in pairs $(p_n,-\overline{p_n})$ due to the symmetry $\overline{\varphi(-\overline{p})}=\varphi(p)$, we have $\prod_n'\frac{|1+p_n^2|}{1+p_n^2}=1$ (where the dash indicates a product over all zeros different from $i$), and thus these convergence factors in the Blaschke product are not needed, which explains the simplification of \eqref{eq:BlaschkeProduct} in comparison to the usual formula for Blaschke products for the upper half plane \cite{Garnett:2007}. We also mention that the potentially infinite product in \eqref{eq:BlaschkeProduct} converges uniformly on compact subsets of $\Cl_+$ thanks to the zeros of $\varphi$ satisfying the Blaschke condition $\sum_n\frac{{\rm Im}\,p_n}{1+|p_n|^2}<\infty$ \cite{Garnett:2007}.

In view of the canonical factorization of $\varphi$, it might seem worthwhile to consider the properties of the localization radius $r$ as a map between the semigroup of symmetric inner functions and the semigroup $[0,\infty]$. In fact, with the help of Lemma~\ref{Lemma:PhiAsRatio-1d} and the smoothing used in the proof of Theorem~\ref{Theorem:CyclicAboveR}~$i)$ above, one can show sub-additivity of the localization radius\footnote{Possibly even additivity $r_{\varphi_1\cdot\varphi_2}=r_{\varphi_1}+r_{\varphi_2}$ holds, i.e. $r$ might be a semigroup homomorphism. We have no definite result on this, but all examples we know comply with the homomorphism property.},
\begin{align}
	\max\{r_{\varphi_1}, r_{\varphi_2}\} \leq r_{\varphi_1\cdot\varphi_2}\leq r_{\varphi_1}+r_{\varphi_2},
\end{align}
for arbitrary symmetric inner functions $\varphi_1,\varphi_2$. (The first inequality in this line follows directly from the definition of the interval subspaces.) Even though $r_\varphi$ can be computed for the elementary factors of the canonical factorization \eqref{eq:PhiFactorization}, i.e. for $\varphi$ a plane wave, a single Blaschke factor, or a singular inner function, this point of view is not efficient for determining $r_\varphi$ in general because $\varphi\mapsto r_\varphi$ is discontinuous. In fact, we will see below convergent infinite products $\varphi=\prod_k \varphi_k$ of symmetric inner functions such that $r_{\varphi_k}=0$ for all $k$, but $r_\varphi$ can take any value in $[0,\infty]$. 

We will therefore work more directly with the canonical factorization of $\varphi$, and recall some more points in this respect. The functions $\varphi$ and $B$ \eqref{eq:PhiFactorization} have precisely the same zeros $p_n$, and we denote their {\em convergence exponent} as
\begin{align}\label{eq:ConvergenceExponent}
	\rho
	:=
	\inf\{\alpha\geq0\,:\,\sum_n |p_n|^{-\alpha}<\infty\}
	\,.
\end{align}
For any $\rho\in[0,\infty]$, there exist Blaschke products with convergence exponent $\rho$, where $\rho=\infty$ means that the series in \eqref{eq:ConvergenceExponent} diverges for all $\alpha>0$. This is in particular the case if the zeros $\{p_n\}$ have a finite limit point.

The factor $S$ is singular at the boundary: If $\varphi$ has an analytic continuation from $\Cl_+$ across some interval $K\subset\Rl$, then $\mu(K)=0$ \cite{RosenblumRovnyak:1994}. Furthermore, any point $t\in\Rl$ with $\mu(\{t\})\neq0$ gives rise to an essential singularity of $S$ at $t$.

The data $x, \rho,\mu$ are uniquely determined by $\varphi$, and we will therefore also denote them by $x_\varphi$, $\rho_\varphi$, $\mu_\varphi$. The sign appearing in \eqref{eq:PhiFactorization} is of no importance for the analysis of the spaces $H_\varphi(I_r)$, and without loss of generality, we can put it to $+1$ in the following.

Our first proposition shows that in case $\varphi$ has a singular part, or the density of its zeros is too high, all interval subspaces $H_\varphi(I)$ are trivial.

\begin{proposition}\label{Proposition:NecessaryConditions}
	Let $\varphi$ be a symmetric inner function.
	\begin{enumerate}
		\item If $\rho_\varphi>1$ or $\mu_\varphi\neq0$, then $r_\varphi=\infty$.
		\item $r_\varphi\geq\frac{1}{2}\,x_\varphi$.
	\end{enumerate}
\end{proposition}
\begin{proof}
	{\em i)} We give a proof by contradiction. Assume $r_\varphi<\infty$, i.e. there exists $0<r<\infty$ such that $H_\varphi(I_r)\neq\{0\}$, which by Lemma~\ref{Lemma:PhiAsRatio-1d} is equivalent to $\varphi=\psi_1/\psi_2$ with non-zero $\psi_1,\psi_2\in H(I_r)$. This has two implications on the structure of $\varphi$: First, all zeros $\{p_n\}$ of $\varphi$ are also zeros of $\psi_1$. Let us denote the potential additional zeros of $\psi_1$ by $\{q_n\}$. As $\psi_1$ is of order~$1$, the convergence exponent of all its zeros $\{p_n\}\cup\{q_n\}$ is at most 1 \cite[Thm.~2.5.18]{Boas:1954}, i.e. we have
	\begin{align}
		\sum_n|p_n|^{-1-\eps}+\sum_n|q_n|^{-1-\eps}<\infty
	\end{align}
	for each $\eps>0$. This implies that $\sum_n|p_n|^{-1-\eps}<\infty$, i.e. $\rho_\varphi\leq1$. In particular, the $p_n$ can have no finite limit point and as a consequence, the Blaschke product $B$ in the factorization of $\varphi$ has a meromorphic extension to all of $\Cl$.
	
	Second, the equation $\varphi=\psi_1/\psi_2$ implies that $\varphi$ must have a meromorphic continuation to the full complex plane. By the previous remark on $B$, and the analyticity of the exponential factor $p\mapsto e^{ipx_\varphi}$, it is clear that this is the case if and only if also $S$ \eqref{eq:PhiFactorization} extends meromorphically to $\Cl$. But according to our earlier remarks on $S$, this is the case if and only if $\mu_\varphi=0$. 
	
	{\em ii)} Let us decompose $\varphi(p)=e^{ipx_\varphi}\cdot \tilde{\varphi}(p)$ into the exponential part and a remainder $\tilde{\varphi}$, i.e. $\varphi(P)=T(x_\varphi)\tilde{\varphi}(P)$. Then
	\begin{align*}
		H_\varphi(I_r)
		&=
		T(-r)T(x_\varphi)\tilde{\varphi}(P)H\cap T(r)H'
		\\
		&\subset
		T(-r+x_\varphi)\tilde{\varphi}(P)H\cap T(r)\tilde{\varphi}(P)H'
		\\
		&=
		\tilde{\varphi}(P)H(I_{-r+x_\varphi,r})
		\,,
	\end{align*}
	because $H'\subset \tilde{\varphi}(P)H'$ and $\tilde{\varphi}(P)$ commutes with translations. If $x_\varphi\geq 2r$, then $I_{-r+x_\varphi,r}$ is empty and $T(a)H(I_{-r+x_\varphi,r})\subset H\cap H'$ for suitable $a$. But $H\cap H'=\{0\}$ by Proposition~\ref{Proposition:CyclicityForHNet}~$ii)$, and thus also $H_\varphi(I_r)=\{0\}$. Hence $r_\varphi\geq\frac{1}{2}\,x_\varphi$.	
\end{proof}

Whereas the second part of this proposition estimates the localization radius from below by showing absence of localized vectors in intervals shorter than the trivial threshold $x_\varphi$ arising from the translation part of $\varphi(P)$, the next result establishes existence of localized vectors in intervals longer than $x_\varphi$, provided $\varphi$ is sufficiently regular.

\begin{proposition}\label{Proposition:SufficientConditions}
	Let $\varphi$ be a symmetric inner function with no singular part, i.e. $\mu_\varphi=0$, and with convergence exponent $\rho_\varphi<1$. Then $r_\varphi=\frac{1}{2}\,x_\varphi$.
\end{proposition}
\begin{proof}
	We denote the zeros of $\varphi(p)=e^{ipx_\varphi}B(p)$ by $\{p_n\}$. As $\rho:=\rho_\varphi<1$, the canonical product of genus 0,
	\begin{align}
		Q_+(p)
		:=
		\prod_n \left(1-\frac{p}{p_n}\right)
		\,,
	\end{align}
	running over all (finitely or infinitely many) zeros of $\varphi$, converges uniformly on compact subsets to an entire function of order $\rho$ and type $0$ \cite[Thm.~2.6.5, Lemma~2.10.13]{Boas:1954}, i.e.  for any $\eps>0$ we find $C_\eps>0$ such that 
	\begin{align}\label{eq:QBound}
		|Q_+(p)|
		\leq
		C_\eps\,e^{\eps |p|^{\rho+\eps}}
		\,,\qquad p\in\Cl
		\,.
	\end{align}
	The same is true for $Q_-(p):=\overline{Q_+(\overline{p})}$, and as the zeros $\{p_n\}$ occur in symmetric pairs $(p_n,-\overline{p_n})$, we have the symmetry $\overline{Q_\pm(-\overline{p})}=Q_\pm(p)$.

	Because of our summability assumption on the $p_n$, we have the following equalities of converging products
	\begin{align*}
		B(p)
		=
		{\prod_n}
		\frac{p-p_n}{p-\overline{p_n}}
		=
		{\prod_n}
		\frac{p_n}{\overline{p_n}}\cdot\frac{1-p/p_n}{1-p/\overline{p_n}}
		=
		\pm\,
		\frac{Q_+(p)}{Q_-(p)}
		\,,
	\end{align*}
	where we have used that $p_n/\overline{p_n}=-1$ for the purely imaginary zeros, and the product $\prod'_n p_n/\overline{p_n}$ taking over the zeros with non-zero real part equals $1$ because they occur in pairs $(p_n,-\overline{p_n})$. Again, $r_\varphi$ is independent of the sign coming from the imaginary zeros, and we may put it to $+1$ without loss of generality.
	
	The growth of the canonical product $Q_\pm$ at infinity is compatible with the growth of elements of $H(I_r)$ off the real axis, but in general, $Q_\pm$ will grow too fast on $\Rl$ to restrict to an element of $L^2(\Rl_+,dp/p)$. To fix this, we introduce another function: Given any $a>0$ and $0<\delta<1$, there exist non-zero, real, continuous functions $\tilde{M}$ with support in $(-a,a)$ such that their Fourier transforms satisfy the bound
	\begin{align}\label{eq:MBound}
		|M(p)|
		\leq
		c\,e^{-\tau|p|^\delta}
		\quad 
		\text{for all {\em real} } p
		\,,
	\end{align}
	with some $c,\tau>0$, see for example \cite{Ingham:1934} for an explicit construction. Because $\tilde{M}$ is real, continuous, and supported in $(-a,a)$, its Fourier transform is entire, and we also have $\overline{M(-\overline{p})}=M(p)$ and $|M(p)|\leq c'\,e^{a|p|}$ for all $p\in\Cl$, for some $c'>0$. 
	
	After these remarks, we set $\psi_\pm(p):=ip\,e^{\pm\frac{i}{2}\,px_\varphi}\,Q_\pm(p)M(p)$, so that 
	\begin{align*}
		\frac{\psi_+(p)}{\psi_-(p)}
		=
		\frac{ip\,e^{\frac{i}{2}\,px_\varphi}\,Q_+(p)M(p)}{ip\,e^{-\frac{i}{2}\,px_\varphi}\,Q_-(p)M(p)}
		=
		e^{ipx_\varphi}\,\frac{Q_+(p)}{Q_-(p)}
		=
		e^{ipx_\varphi}\,B(p)
		=
		\varphi(p)\,,
	\end{align*}
	as required in Lemma~\ref{Lemma:PhiAsRatio-1d}. We have to show that $\psi_\pm$ lies in $H(I_r)$ for suitable $r$. It is clear that $\psi_\pm$ is entire, vanishes at $p=0$, and satisfies the reality condition $\overline{\psi_\pm(-\overline{p})}=\psi_\pm(p)$ because each factor in its defining equation does. 
	
	To show that $\psi_\pm$ restricts to a function in $L^2(\Rl_+,dp/p)$, we choose $\delta$ in the definition of $M$ such that $1>\delta>\rho+\eps$, where $\eps$ is taken from the bound on $Q_+$, and estimate using \eqref{eq:MBound} for real $p>0$
	\begin{align*}
		|\psi_\pm(p)|
		=
		|e^{\pm\frac{i}{2}\,px_\varphi}pM(p)Q(p)|
		&\leq
		p\cdot\,c\,e^{-\tau\,p^\delta}\cdot C_\eps\,e^{\eps\,p^{\rho+\eps}}
		\,.
	\end{align*}
	Because of our choice of $\delta$, the right hand side converges rapidly to zero as $p\to\infty$, and in view of the explicit factor $p$, it follows that $\psi_\pm\in L^2(\Rl_+,dp/p)$.
	
	To estimate the exponential type of $\psi_\pm$, note that we have for all $p\in\Cl$
	\begin{align*}
		|\psi_\pm(p)|
		=
		|e^{\pm\frac{i}{2}\,px_\varphi}\,p\,M(p)Q(p)|
		\leq
		e^{\frac{x_\varphi}{2}|{\rm Im}\,p|}
		\cdot
		|p|\cdot
		c'\,e^{a|p|}\cdot C_\eps\, e^{\eps|p|^{\rho+\eps}}
		\leq
		c
		\,e^{(\frac{x_\varphi}{2}+a+\eps)|p|}
		\,,
	\end{align*}
	with a suitable constant $c$. As we are still free to choose $a>0$ and $\eps>0$ as small as we like, we see that $\psi_\pm\in H(I_{\frac{x_\varphi}{2}})$ by the characterization in Proposition~\ref{Proposition:CharacterizationOfHI}~$iv)$. According to Lemma~\ref{Lemma:PhiAsRatio-1d}, this implies that $H_\varphi(I_r)$ is cyclic for all $r>\frac{x_\varphi}{2}$, i.e. $r_\varphi\leq \frac{x_\varphi}{2}$. But $r_\varphi\geq\frac{x_\varphi}{2}$ in general by Proposition~\ref{Proposition:NecessaryConditions}~{\em ii)}, so we have $r_\varphi=\frac{x_\varphi}{2}$.
\end{proof}

Note that it follows from these two propositions and the structure of the symmetric inner functions \eqref{eq:PhiFactorization} that to each $r\in[0,\infty]$, there exists a symmetric inner function $\varphi$ such that $r_\varphi=r$. However, Propositions~\ref{Proposition:NecessaryConditions} and \ref{Proposition:SufficientConditions} do not cover the case of convergence exponent $\rho_\varphi=1$. In this case, $r_\varphi$ depends on more detailed properties of $\varphi$ than those described by $x_\varphi$, $\rho_\varphi$, $\mu_\varphi$, as we will now demonstrate with a family of examples, which realize any localization radius $r_\varphi\in(0,\infty)$, and still satisfy $\mu_\varphi=0$, $x_\varphi=0$.

\begin{example}
	Let $\nu,q>0$. The function $\varphi$ defined by
	\begin{align*}
		\varphi(p):=\frac{s_{\nu,q}(p)}{s_{\nu,-q}(p)}
		\,,\qquad
		s_{\nu,q}(p):=i\sin(\nu p-iq)
	\end{align*}
	is a symmetric inner function on the upper half plane with $x_\varphi=0$, $\mu_\varphi=0$, $\rho_\varphi=1$, independent of $\nu$ and $q$. Its minimal localization radius is $r_\varphi=\nu$.
\end{example}
\begin{proof}
	The zeros of $\varphi$ are $p_n=\nu^{-1}(iq+\pi n)$, $n\in\Zl$, which yields the convergence exponent $\rho_\varphi=1$. We also see that the $p_n$ satisfy the Blaschke summability condition $\sum_n\frac{{\rm Im}\,p_n}{1+|p_n|^2}<\infty$, so that the corresponding Blaschke product $B$ is convergent. We have
	\begin{align*}
		B(p)
		&=
		\prod_{n\in\Zl}\frac{p-\frac{iq+\pi n}{\nu}}{p-\frac{-iq+\pi n}{\nu}}
		=
		\frac{\nu p-iq}{\nu p+iq}\prod_{n=1}^\infty\frac{1-\frac{(\nu p-iq)^2}{\pi^2n^2}}{1-\frac{(\nu p+iq)^2}{\pi^2n^2}}\,,
	\end{align*}
	which is seen to coincide with $\varphi(p)$, $p\in\Cl_+$, by taking into account the product formula $\sin(z)=z\prod_{n=1}^\infty (1-\frac{z^2}{\pi^2n^2})$. Thus $\varphi=B$ is a Blaschke product for the upper half plane, i.e. an inner function with $x_\varphi=0$, $\mu_\varphi=0$. The symmetry property $\overline{\varphi(-\overline{p})}=\varphi(p)$ is easily verified.
	
	To determine the localization radius $r_\varphi$, let $f\in C_{c,\Rl}^\infty(I_r)$, so that $\widehat{f'}\in H(I_r)$. As the inverse Fourier transform of $s_{\nu,\pm q}$ is a real linear combination of delta distributions at $x=\pm\nu$, the inverse Fourier transform of $s_{\nu,\pm q}\cdot\widehat{f'}$ has support in $I_{r+\nu}$. Thus $s_{\nu,\pm q}\cdot\widehat{f'}\in H(I_{r+\nu})$ and $\varphi=(s_{\nu,q}\cdot\widehat{f'})/(s_{\nu,-q}\cdot\widehat{f'})$, which by Lemma~\ref{Lemma:PhiAsRatio-1d} implies $H_\varphi(I_{r+\nu})\neq\{0\}$. As we can choose $r>0$ as small as we like, we conclude $r_\varphi\leq\nu$.
	
	Conversely, assume that $r_\varphi<\nu$. Then there exists $r<\nu$ such that $\varphi=\psi_1/\psi_2$ with $\psi_1,\psi_2\in H(I_r)$ non-zero; these functions may be represented as $\psi_k(p)=-ip\widehat{g}_k(p)$ for some real $g_k\in L^2(I_r,dx)$, $k=1,2$. Then we have $s_{\nu,q}\,\widehat{g_2}=s_{\nu,-q}\,\widehat{g_1}$, which implies by inverse Fourier transformation
	\begin{align*}
		e^q\,g_2(x-\nu)-e^{-q}g_2(x+\nu)
		=
		e^{-q}\,g_1(x-\nu)-e^{q}g_1(x+\nu)
		\,,\qquad x\in\Rl\,.
	\end{align*}
	As $r<\nu$, the supports of the functions shifted by $\nu$ and $-\nu$ are disjoint, and comparison of the two parts gives $g_2=e^{-2q}g_1$ and $g_2=e^{2q}g_1$, which is impossible for $q\neq0$, $g_1,g_2\neq0$. Thus $r_\varphi=\nu$.
\end{proof}

We also point out that in case $\varphi$ is $\gamma$-invariant, the situation simplifies.

\begin{lemma}\label{Lemma:GammaSymmetricPhi-1d}
	Let $\varphi$ be a symmetric inner function with $\gamma(\varphi)=\varphi$.
	\begin{enumerate}
		\item Either $\rho_\varphi=0$ or $\rho_\varphi=\infty$, and thus $r_\varphi=\infty$ in the latter case.
		\item If $\rho_\varphi=0$, the following statements are equivalent:
		\begin{enumerate}
			\item[a)] $r_\varphi=0$.
			\item[b)] $\mu_\varphi=0$ and $x_\varphi=0$, i.e. $\varphi$ is a (finite) Blaschke product.
			\item[c)] $\varphi\circ\exp$ extends to a bounded analytic function on a strip of the form $-\eps<{\rm Im}\,\te<\pi+\eps$ for some $\eps>0$.
		\end{enumerate}
	\end{enumerate}
\end{lemma}
\begin{proof}
	{\em i)} In view of the symmetry $\varphi(-p^{-1})=\varphi(p)$, the zeros $p_n$ of $\varphi$ occur in pairs of the form $(p_n,-p_n^{-1})$. Hence, if $|p_n|^{-\alpha}$ is summable over $n$ for some $\alpha>0$, also $|p_n|^\alpha$ is summable over~$n$. This is possible if and only if the number of zeros is finite, in which case we have $\rho_\varphi=0$. If the number of zeros is infinite, we necessarily have $\rho_\varphi=\infty$ and, by Proposition~\ref{Proposition:NecessaryConditions}~{\em i)}, $r_\varphi=\infty$. 
		
	{\em ii)} Since $\rho_\varphi=0<1$, we know from Proposition~\ref{Proposition:NecessaryConditions} and Proposition~\ref{Proposition:SufficientConditions} that $r_\varphi=0$ is equivalent to $\mu_\varphi=0$, $x_\varphi=0$, i.e. $a)\Leftrightarrow b)$. According to {\em i)}, the function $\varphi\circ\exp$ has only finitely many zeros in the strip $\Strip_\pi$. Assuming {\em b)}, we have no singularities on the boundary of $\Strip_\pi$ by $\mu_\varphi=0$, and by $\varphi(e^{-\te})=\varphi(e^\te)^{-1}$ and $\varphi(e^{\te+i\pi})=\overline{\varphi(e^\te)}$ and the finite number of zeros, we find an analytic continuation to the broader strip $-\eps<{\rm Im}\,\te<\pi+\eps$ if $0<\eps<\min\{{\rm Im}\,\log p_n\}$. As $\varphi$ is a finite Blaschke product, $\varphi(e^{\te+i\la})\to\pm1$ as $\te\to\pm\infty$, uniformly in $\la\in[0,\pi]$. Thus $\varphi\circ\exp$ is also bounded on the extended strip, i.e. we have shown $b)\Rightarrow c)$. For the reverse conclusion, we first note that the analytic continuation to a 
broader strip requires $\mu_\varphi=0$ as before, and $\varphi\circ\exp$ will be bounded on the broader strip if and only if the translation part $\te\mapsto e^{ix_\varphi e^\te}$ is. But one easily checks that this function is bounded on $-\eps<{\rm Im}\,\te<0$ only if $x_\varphi=0$. Thus also $c)\Rightarrow b)$ holds.
\end{proof}

The significance of the extended analyticity requirement $c)$ is discussed in our conclusions.

\subsection{Cyclicity of double cone subspaces}

In this section we investigate the massive irreducible situation in two dimensions, with representation $U_m$, $m>0$.

Again we first characterize the local subspaces of the net with $V=1$. To this end, we use the time reflection $\Gamma=JZ$, and the Hamiltonian $\om$, multiplying (in the momentum picture) with $\om_m(p_1)$. With these operators, we associate with any $\psi\in\Hil$ the two vectors
\begin{align}\label{eq:psi+-MT}
	\psi_+
	&:=
	\frac{1}{2}(1+\Gamma)\psi\,,\qquad
	\psi_-
	:=
	\frac{1}{2i\om}(1-\Gamma)\psi\,,\qquad
	\Rightarrow\quad
	\psi=\psi_++i\om\psi_-\,,
\end{align}
related to the Cauchy data formulation of the solutions of the Klein Gordon equation. The $\psi_\pm$ depend real linearly and continuously on $\psi$ (since $\|\om^{-1}\|\leq m^{-1}<\infty$) and are invariant under $\Gamma$ since $\Gamma$ and $\om^{-1}$ commute.

\begin{proposition}\label{Proposition:CharacterizationOfHO}
	Let $(H,T_m)$ be the two-dimensional irreducible standard pair with representation $U_m$, $m>0$, and $\psi\in\Hil = L^2(\Rl,\om_m(p_1)^{-1}dp_1)$. The following are equivalent:
	\begin{enumerate}
		\item $\psi\in H^m(\OO_r)$.
		\item The inverse Fourier transforms of $\psi_\pm$ \eqref{eq:psi+-MT} have support in $I_r$.
		\item The functions $\psi_\pm$ \eqref{eq:psi+-MT} extend to entire analytic functions such that $\psi_\pm(-\overline{p_1}) = \overline{\psi_\pm(p_1)}$ and $|\psi_\pm(p_1)| \leq C(1+|p_1|)^N  e^{r|{\rm Im}\,p_1|}$,  $p_1\in\mathbb C$.
		\item The functions $\psi_\pm$ \eqref{eq:psi+-MT} extend to entire analytic functions such that $\psi_\pm(-\overline{p_1}) = \overline{\psi_\pm(p_1)}$ and of exponential type at most $r$, namely for all $\eps>0$ there exists $C_\eps>0$ such that $|\psi_\pm(p_1)| \leq C_\varepsilon  e^{(r+\varepsilon)|p_1|}$,  $p_1\in\mathbb C$.
	\end{enumerate}
	The subspace
	\begin{align}\label{eq:HISmooth+-MT}
		\{\psi\,:\,\check{\psi}_\pm\in C_{c,\Rl}^\infty(I_r)\}
		\subset 
		H^m(\OO_r)
	\end{align}
	is cyclic.
\end{proposition}

We also give a characterization in the lightray picture.
\begin{corollary}\label{Corollary:CharacterizationOfHO-lightray}
	In the lightray representation, let $\psi\in\Hil= L^2(\Rl_+,dp/p)$. The following are equivalent:
	\begin{enumerate}
		\item $\psi\in H^m(\OO_r)$
		\item There exist $\psi^\pm\in\Hil$, $\overline{\psi^\pm(\frac{1}{p})}=\overline{\psi_\pm(-p)}=\psi^\pm(p)$, $p>0$, which are restrictions to $\Rl_+$ of functions analytic on $\Cl\backslash\{0\}$ and satisfy for any $\eps>0$ the bound
		\begin{align}
			|\psi^\pm(p)|\leq C_{\eps,\pm}\,e^{(\frac{mr}{2}+\eps)|p-\frac{1}{p}|}
			\,,\qquad
			p\in\Cl\backslash\{0\}\,,
		\end{align}
		for some $C_{\eps,\pm}>0$, such that
		\begin{align}\label{eq:CauchyRepPsi-LightRay}
			\psi(p)
			=
			\psi^+(p)+i\,(p+\tfrac{1}{p})\cdot\psi^-(p)
			\,,\qquad
			p>0\,.
		\end{align}
	\end{enumerate}
\end{corollary}

Both these results are proven in Appendix~\ref{Section:Appendix2d}.

We begin with the analogue of Lemma~\ref{Lemma:PhiAsRatio-1d}, and the proof of Theorem~\ref{Theorem:CyclicAboveR}~$ii)$.

\begin{lemma}\label{Lemma:PhiAsRatio-2d}
		Let $\varphi$ be a symmetric inner function, and $r>0$. Then $H^m_\varphi(\OO_r)\neq\{0\}$ if and only if there exist non-zero $\psi_1,\psi_2\in H^m(\OO_r)$ such that $\varphi(\frac{mp}{2})=\psi_1(p)/\psi_2(p)$, ${\rm Im}\,p>0$.
\end{lemma}
\begin{proof}
	By Lemma~\ref{Lemma:BasicNetPropertiesForHV}~$i)$, $H^m_\varphi(\OO_r)=H^m(\OO_r)\cap \varphi(P_+)H^m(\OO_r)$, with $P_+=\frac{m}{2}P$. Thus $H^m_\varphi(\OO_r)\neq\{0\}$ is equivalent to the condition that there exist non-zero $\psi_1,\psi_2\in H^m(\OO_r)$ such that $\varphi(\frac{mp}{2})=\psi_1(p)/\psi_2(p)$, ${\rm Im}\,p>0$.
\end{proof}

	\noindent{\em Proof of Theorem~\ref{Theorem:CyclicAboveR}~$ii)$}. Let $r>r_{m,\varphi}$, and pick $r'$ such that $r_{m,\varphi}<r'<r$. By definition of $r_{m,\varphi}$, we have $H^m_\varphi(\OO_{r'})\neq\{0\}$, and by the preceding Lemma, this implies $\varphi=\psi_1/\psi_2$ with certain non-zero $\psi_1,\psi_2\in H^m(\OO_{r'})$. Let $f_\pm\in C_{c,\Rl}^\infty(I_\eps)$, with $\eps:=r-r'>0$, so that $f:=\hat{f}_++i\om \hat{f}_-\in H^m(\OO_\eps)$ (Prop.~\ref{Proposition:CharacterizationOfHO}). A short calculation shows that the product function $\psi_k \cdot f$, $k=1,2$, decomposes in
	\begin{align*}
		(\psi_k f)_+
		&=
		\psi_{k,+}f_+-\om^2\psi_{k,-}f_-
		\,,\qquad
		(\psi_k f)_-
		=
		\psi_{k,-}f_++\psi_{k,+}f_-\,.
	\end{align*}
	Making use of Prop.~\ref{Proposition:CharacterizationOfHO}~$iv)$, we see that $\psi_k f\in H^m(\OO_{r})$. As in the proof of Theorem~\ref{Theorem:CyclicAboveR}~$i)$, this implies $\psi_1 f\in H^m_\varphi(\OO_{r})$. Let $\eta$ be in the orthogonal complement of the complex linear span of all vectors of the form $\psi_1 f$, then $\check{f}_{1,\pm},\check{f}_{2,\pm}\in C_{c,\Rl}^\infty(I_\eps)$,
	\begin{align*}
		0
		&=
		\langle\eta,\psi_1\cdot(f_1+if_2)\rangle
		=
		\langle\overline{\psi_1}\cdot \eta,(f_1+if_2)\rangle
		\,,
	\end{align*}
	which implies $\overline{\psi_1}\cdot\eta=0$ by the cyclicity \eqref{eq:HISmooth+-MT}. In view of the representation \eqref{eq:psi+-MT}, $\psi_1\neq0$ is analytic on $\Cl\backslash\{\pm i[m,\infty)\}$, and thus has only isolated zeros on $\Rl$. Hence $\eta=0$. {\hfill $\square$ \\[2mm] \indent}

The influence of the analytic properties of $\varphi$ on the two-dimensional net $H^m_\varphi$ is different from the one on the one-dimensional case $H_\varphi$. The singular functions $\varphi=S_{\mu}$ with measure $\mu=a\delta_0$ are $\varphi(p)=e^{-iap^{-1}}$, i.e. $\varphi(\frac{m}{2}P_+)=T'(-a)$, which do not lead to absence of localized vectors on all scales as in Proposition~\ref{Proposition:NecessaryConditions}. This translation parameter can be recovered from $\varphi$ as $a=x_{\gamma(\varphi)}$ and leads to the natural geometric restriction on $r_{m,\varphi}$.

\begin{proposition}
	Let $\varphi$ be a symmetric inner function for the upper half plane (in the lightray representation).
	\begin{enumerate}
		\item If $\mu_\varphi\neq a\cdot\delta_0$ for some $a\geq0$, then $r_{m,\varphi}=\infty$.
		\item $r_{m,\varphi}\geq\frac{1}{2}\max\{x_\varphi,x_{\gamma(\varphi)}\}$.
	\end{enumerate}
\end{proposition}
\begin{proof}
	$i)$ This proof works as the proof of Proposition~\ref{Proposition:NecessaryConditions}~$i)$ with the only difference that by Corollary~\ref{Corollary:CharacterizationOfHO-lightray}, functions in $H^m(\OO_r)$ are analytic in $\Cl\backslash\{0\}$, and therefore only essential singularities different from zero (corresponding to points $p\neq0$ in $\supp\mu_\varphi$) are excluded.
	
	$ii)$ Assuming $\mu_\varphi=a\delta_0$, we can decompose $\varphi$ as $\varphi(p)=\pm e^{ix_\varphi p}e^{-ix_{\gamma(\varphi)}p^{-1}}B(p)$ with some Blaschke product for the upper half plane, i.e. $V=\varphi(\frac{m}{2}P)=T_m(x)V_0$, where $V_0=\pm B(\frac{m}{2}P)\in\E(H,T_m)$ and $x\in\Rl^2$ is given by $x_+=x_\varphi$, $x_-=x_{\gamma(\varphi)}$. We have
	\begin{align*}
		H^m_\varphi(\OO_r)
		&=
		T_m((0,-r)+x)V_0H\cap T_m(0,r)H'
		\\
		&\subset 
		H^m((-W+(0,r)+x)\cap(W-(0,r))
	\end{align*}
	We observe that $W-({0\atop r})+x$ is disjoint from $-W+({0\atop r})$ if $r<\frac{1}{2}x_\varphi$ or $r<\frac{1}{2}x_{\gamma(\varphi)}$. This implies that the intersection of the corresponding standard spaces is $\{0\}$: In fact, the corresponding second quantized von Neumann algebras $\A(W-({0\atop r})+x)$ and $\A(W+({0\atop r}))$ have only multiples of the identity in common \cite{Kuckert:2000}, but any non-zero $h\in H^m((-W+(0,r)+x)\cap(W-(0,r))$ would give rise to non-trivial operators in that intersection. Hence we conclude $H^m_\varphi(\OO_r)=\{0\}$ for $r<\frac{1}{2}x_\varphi$ or $r<\frac{1}{2}x_{\gamma(\varphi)}$, i.e. $r_{m,\varphi}\geq\frac{1}{2}\max\{x_\varphi,x_{\gamma(\varphi)}\}$.
\end{proof}
	
In view of the general inequality $r_{m,\varphi}\leq\min\{r_\varphi,\,r_{\gamma(\varphi)}\}$ (Prop.~\ref{Proposition:1dvs2dNets}) and the results on the one-dimensional case, we see that if for example $\varphi$ is a Blaschke product with convergence exponent $\rho_\varphi<1$ or $\rho_{\gamma(\varphi)}<1$, we get $r_{m,\varphi}=0$, i.e. $H^m_\varphi(\OO_r)$ is cyclic for all $r>0$. In this case, also at least one of the nets $I\mapsto H_\varphi(I)$ or $I\mapsto H_{\gamma(\varphi)}(I)$ has cyclic subspaces for intervals of any size ($r_\varphi=0$ or $r_{\gamma(\varphi)}=0$), whereas the other net might or might not have trivial subspaces for all intervals. But also the other extreme case, i.e. $r_\varphi=r_{\gamma(\varphi)}=\infty$, but $r_{m,\varphi}=0$, is realized by suitable $\varphi$, as we show show next.

\begin{proposition}
	There exist symmetric $\gamma$-invariant Blaschke products $B$ for the upper half plane such that $r_B=r_{\gamma(B)}=\infty$ and $r_{m,B}=0$.
\end{proposition}
\begin{proof}
	Let $B$ be a symmetric $\gamma$-invariant Blaschke product for the upper half plane, with infinitely many zeros. According to Lemma~\ref{Lemma:GammaSymmetricPhi-1d}, we then have $\rho_B=\infty$, $r_B=r_{\gamma(B)}=\infty$. 

	To show that $r_{m,B}=0$, we have to demonstrate that for any $r>0$, $B$ can be represented in the form $B(\frac{mp}{2})=\psi_1(p)/\psi_2(p)$ with $\psi_1,\psi_2\in H^m(\OO_r)$. In view of the $\gamma$-invariance of $B$, this will be the case if and only if $B(\frac{mp}{2})=\psi_{1,\pm}(p)/\psi_{2,\pm}(p)$, or, in the momentum representation, where $p=\frac{1}{m}( \om_m(p_1)+p_1)$,
	\begin{align*}
		B'(p_1)
		:=
		B(\tfrac{1}{2}(\om_m(p_1)+p_1))
		=
		\frac{\psi_{1,\pm}(p_1)}{\psi_{2,\pm}(p_1)}
		\,.
	\end{align*}
	By the same line of argument as in the proof of Corollary~\ref{Corollary:CharacterizationOfHO-lightray-appendix}, one sees that $B'$ is analytic in the upper half plane. As it is also non-singular at the boundary, it follows that $B'$ is also a Blaschke product for the upper half plane. We have shown in the proof of Proposition~\ref{Proposition:SufficientConditions} that given any $r>0$, a Blaschke product of convergence exponent strictly smaller one can be written as the ratio of two entire functions $f,g$ of exponential type $r$, which decay rapidly to zero as $p_1\to\pm\infty$ on the real line, and satisfy the reality condition $\overline{f(-\overline{p_1})}=f(p_1)$, $\overline{g(-\overline{p_1})}=g(p_1)$. These properties precisely match the characterization of $\psi_{1,\pm},\psi_{2,\pm}\in H^m(\OO_r)$ (see Prop.~\ref{Proposition:CharacterizationOfHO}~$iv)$).
	
	To conclude the argument, we therefore have to construct $B$ in such a way that $\rho_{B'}<1$. To this end, we consider the zeros $p_n:=n^{-\beta}e^{i/n}$, $n\in\Nl$, $\beta>0$, which lie in the upper half plane and satisfy the Blaschke condition $\sum_n\frac{{\rm Im}\,p_n}{1+|p_n|^2}\leq\sum_n n^{-\beta-1}<\infty$. Denoting the corresponding Blaschke product $b$, we set $B(p):=b(p)\overline{b(-\overline{p})}\overline{b(\overline{p}^{-1})}b(-p^{-1})$, which is a symmetric $\gamma$-invariant Blaschke product for the upper half plane.
	
	The zeros $p_{1,n}$ of the function $\varphi(p_1):=b(\tfrac{1}{2}(\om_m(p_1)+p_1))$ are $p_{1,n}=p_n-\frac{m^2}{4p_n}$, and we have for sufficiently large $n$
	\begin{align*}
		\left|p_n-\frac{m^2}{4p_n}\right|^{-1}
		&=
		\left|n^{-\beta}-\frac{m^2n^\beta\,e^{-2i/n}}{4}\right|^{-1}
		\leq
		\left|n^{-\beta}-\frac{m^2n^\beta}{4}\right|^{-1}
		\leq
		\frac{1}{\frac{m^2n^\beta}{4}-1}\,,
	\end{align*}
	which shows $\sum_n|p_{1,n}|^{-\alpha}<\infty$ for $\alpha>\beta^{-1}$, i.e. $\rho_\varphi\leq\beta^{-1}$. Since the zeros  $p_{1,n}=p_n-\frac{m^2}{4p_n}$ are essentially unchanged under the replacements $p\to-p$ and $p\to p^{-1}$, we also find convergence exponent at most $\beta^{-1}$ for the other factors, and thus $\rho_{B'}\leq\beta^{-1}$. For $\beta>1$ we have thus constructed a Blaschke product with the claimed properties. 
\end{proof}

\section{Discussion and Conclusions}\label{Section:Conclusions}

The nets of standard subspaces that we have constructed here can be promoted to nets of von Neumann algebras by second quantization. In fact, let $\F^\pm_\Hil$ denote the Bose ($+$) or Fermi ($-$) Fock space over $\Hil$, and for any closed real-linear subspace $H\subset\Hil$, let $\M_\pm(H)\subset\B(\F^\pm_\Hil)$ denote the von Neumann algebra generated by the Weyl operators $W(h)$, $h\in H$ ($+$) or the Fermi field $a^*(h)+a(h)$, $h\in H$ ($-$), respectively. Then $H\mapsto\M_\pm(H)$ maps standard subspaces in $\Hil$ to von Neumann algebras in $\B(\F_\Hil^\pm)$ which have the Fock vacuum $\Om$ as a cyclic and separating vector. The mapping $H\mapsto\M_\pm(H)$ preserves inclusions, maps symplectic complements to commutants (twisted commutants in the Fermi case), and also the Poincar\'e symmetries transport to the algebraic level by second quantization \cite{LeylandsRobertsTestard:1978,Derezinski:2004}.

If $(H,T)$ is a one- or two-dimensional non-degenerate standard pair, we thus have the ``free'' nets $I\mapsto\M_\pm(H(I))$, and for any endomorphism $V\in\E(H,T)$, the subnets
\begin{align*}
	I\longmapsto\M_\pm(H_V(I))\subset\M_\pm(H(I))\,.
\end{align*}

In particular, considering irreducible $(H,T)$ with mass $m>0$, we get for any symmetric inner function $\varphi$ the nets $\OO\mapsto\M_\pm(H^m_\varphi(\OO))$ on $\Rl^2$. 

A family of nets indexed by certain symmetric inner functions appears also in a different setting. In fact, if $\varphi$ is $\gamma$-invariant, then it satisfies all properties of a {\em scattering function} \cite{Lechner:2003} (usually formulated in the rapidity representation), i.e. $\varphi$ can be seen as the kernel of an elastic two-particle S-matrix which is unitary, hermitian analytic, and crossing symmetric. One can then formulate the inverse scattering problem, aiming at the construction of a net of von Neumann algebras describing a quantum field theory on two-dimensional Minkowski space, with particles of mass $m>0$ whose collision theory is governed by the factorizing S-matrix derived from $\varphi$. Starting from concepts of the form factor program \cite{Smirnov:1992,BabujianFoersterKarowski:2006}, and Schroer's invention of polarization-free generators \cite{Schroer:1997-1,Schroer:1999}, such a construction was carried out in an operator-algebraic setting in \cite{Lechner:2003,Lechner:2008}, see also \cite{BostelmannLechnerMorsella:2011,BostelmannCadamuro:2012,BischoffTanimoto:2013,LechnerSchutzenhofer:2013} for extensions of this construction and related work.

In comparison to the construction in the present paper, the resulting structure can be summarized as follows: For any $\gamma$-invariant symmetric inner function $\varphi$, one obtains a net $\OO\mapsto\M^\varphi(\OO)$ of von Neumann algebras, indexed by double cones $\OO\subset\Rl^2$. This net is fully Poincar\'e covariant, and for each $\varphi$, there exists a minimal length $r^\varphi$ such that $\M^\varphi(\OO_r)$ has the vacuum vector as a cyclic vector if $r>r^\varphi$. 
As $\varphi$ defines the S-matrix of the net, $\M^{\varphi_1}$ and $\M^{\varphi_2}$ are not equivalent if $\varphi_1\neq\varphi_2$. For the constant function $\varphi(p)=1$, one obtains the free net $\M^1(\OO)=\M_+(H(\OO))$ generated from the two-dimensional irreducible standard pair with mass $m>0$ on the Bose Fock space. 

For general $\varphi$, the representation space can conveniently be chosen as $\F^\pm_\Hil$ if $\varphi(0)=\pm1$ (these are the only possibilities for $\gamma$-invariant $\varphi$). However, the algebras $\M^\varphi(\OO)$ do {\em not} coincide with $\M_\pm(H^m_\varphi(\OO))$ --- for example, the net $\M^\varphi$ is dual, fully Poincar\'e covariant, has non-trivial S-matrix, and a one-particle modular structure independent of $\varphi$. All these properties are not realized by the net $\OO\mapsto\M_\pm(H^m_\varphi(\OO))$. Rather, one can view $\M^\varphi(\OO)$ as a ``deformed second quantization'' of the subspace $H(\OO)$, i.e. the standard one particle subspaces can be recovered from $\M^\varphi(\OO)$ independent of $\varphi$, but the second quantization functor is modified in a $\varphi$-dependent manner.

Despite these differences, we strongly believe that there is a deeper connection between endomorphisms of standard pairs and scattering functions. Partial relations between the two scenarios have already been found in \cite{Tanimoto:2011-1,BischoffTanimoto:2013,LechnerSchlemmerTanimoto:2013}, but a complete understanding of this connection has not been reached yet. 

One can also compare the inverse scattering construction with the endomorphism subnet construction with respect to the minimal localization length. The extended analyticity requirement in Lemma~\ref{Lemma:GammaSymmetricPhi-1d}~$ii)\,c)$ comes from the inverse scattering construction, where scattering functions satisfying it were called {\em regular} \cite{Lechner:2008}. It was shown there that one has $r^\varphi<\infty$ for regular $\varphi$, and $r^\varphi=0$ for regular $\varphi$ with $\varphi(0)=-1$. Lemma~\ref{Lemma:GammaSymmetricPhi-1d} reproduces a similar situation also in the situation considered here\footnote{In the context of short distance scaling limits, the class of symmetric inner functions which are finite Blaschke products plays a distinguished role \cite{BostelmannLechnerMorsella:2011}. While in the von Neumann algebraic situation, the existence of interval-localized observables is currently open for non-constant $\varphi$, it is interesting to note that in the subspace picture used here, this question has a simple affirmative answer: For finite Blaschke products, we have $r_{m,\varphi}=r_\varphi=r_{\gamma(\varphi)}=0$.}.

Dropping the requirement $\gamma(\varphi)=\varphi$, one can however easily produce examples of symmetric inner functions $\varphi$ which violate regularity but still satisfy $r_\varphi=0$. This is not surprising when one recalls the background of the regularity condition: Similar to the procedure in the present paper, the wedge algebra $\M^\varphi(W_2)$ of the net $\M^\varphi$ can be constructed directly, and the double cone algebras are defined as appropriate intersections, which by duality is the same as the relative commutant of the inclusion $\M^\varphi(W_2+x)\subset\M^\varphi(W_2)$, $x\in W_2$. The method to guarantee that these intersections have the vacuum as a cyclic vector (and, in particular, are non-trivial) \cite{BuchholzLechner:2004,Lechner:2008}, was to verify the split property \cite{DoplicherLongo:1984} of these inclusions by checking the modular nuclearity condition 
\cite{BuchholzDAntoniLongo:1990-1}. While it seems difficult to prove modular nuclearity without regularity of $\varphi$ \cite{Lechner:2008}, modular nuclearity or the split property are certainly not necessary for the inclusion $\M^\varphi(W_2+x)\subset\M^\varphi(W_2)$ to have a large relative commutant.

Comparing to the situation in the present paper, we note that inclusions of closed real linear subspaces $K\subset H$, with $H'\cap H=\{0\}$, are always ``normal'' in the sense that $(K'\cap H)'\cap H=K$, i.e. each intersection coincides with its double relative symplectic complement, and the relative symplectic complement of $K\subset H$ is trivial if and only if $K=H$. As second quantization is an isomorphism of orthocomplemented lattices, also inclusions of second quantization von Neumann algebras $\M_\pm(K)\subset\M_\pm(H)$, $K\subset H$, are always normal. Thus the relative commutant of $\M_\pm(K)$ in $\M_\pm(H)$ is non-trivial if and only if $K\neq H$. This situation is quite different for general non-normal inclusions of von Neumann algebras, which have a less accessible relative commutant. In this general setting, a strong condition on an inclusion of von Neumann algebras is to be (quasi-)split, which implies normalcy in the factor case \cite{DoplicherLongo:1984}.

However, these conditions are too strong for certain situations like the one-dimensional case considered here. In fact, consider for some $a>0$ the interval $I=(-a,0)$ and its subspace $H_V(I)=T(-a)VH\cap H'\subset H'$. If the second quantization of the inclusion $H_V(I)'\cap H'\subset H'$ was split, then the same would hold in particular for $V=1$, because $H(I)'\subset H_V(I)'$. But for $V=1$ it is known that the inclusion is not split; one can for example falsify certain one-particle modular compactness criteria, which are necessary conditions for split in this setting. 

On the other hand, in the two-dimensional massive irreducible situation, it is known that the second quantization of the inclusion $T(-a)H'\subset H'$ is split for $a\in W_2$. But for suitable endomorphisms $V$, the second quantization of the inclusion $T(-a)VH'\cap H'\subset H'$ is again not split, but still has a standard relative commutant.

These examples illustrate the well-known fact that the split property is too strong as a condition for large relative commutants, and weaker manageable conditions are needed. We hope to come back to these questions in the future.

\appendix
\section{Appendix}\label{Section:Appendix}

In this appendix we give the proofs of Lemma~\ref{Lemma:CharacterizeH} (reprinted here as Lemma~\ref{Lemma:CharacterizeH-Appendix}), Proposition~\ref{Proposition:CharacterizationOfHI} (as Proposition~\ref{Proposition:CharacterizationOfHI-Appendix}), Proposition~\ref{Proposition:CharacterizationOfHO} (as Proposition~\ref{Proposition:CharacterizationOfHO-Appendix}), and Corollary~\ref{Corollary:CharacterizationOfHO-lightray} (as Corollary~\ref{Corollary:CharacterizationOfHO-lightray-appendix}).

\begin{lemma}\label{Lemma:CharacterizeH-Appendix}
		In the rapidity representation, $H=\{\psi\in {\mathbb H}^2(\Strip_{\pi})\,:\,\overline{\psi(\te+i\pi)}=\psi(\te)\text{ a.e.}\}$.
\end{lemma}
\begin{proof}
	Let $\psi\in {\mathbb H}^2(\Strip_{\pi})$. Then, for any $f\in L^2(\Rl,d\te)$, the function $\Rl\ni t \mapsto\langle f,\Delta^{it}\psi\rangle$ extends to an analytic function on the strip $\Strip_{-1/2}$ bounded by $\|f\|\|\psi\|$ thanks to unitarity of $\Delta^{it}$ and the three lines theorem. Thus $t\mapsto \Delta^{it}\psi$ is weakly analytic, which implies \cite[Thm.~VI.4]{ReedSimon:1972} that it is also strongly analytic. Thus $\psi\in\dom\Delta^{1/2}$. The condition $\overline{\psi(\te+i\pi)}=\psi(\te)$ for almost all $\te\in\Rl$ then gives $J\Delta^{1/2}\psi=\psi\Rightarrow\psi\in H$, i.e. we have shown the inclusion ``$\supset$'' in the claimed equality.
	
	To establish also the inclusion ``$\subset$'', pick $\psi\in H$ and $f\in C_c^\infty(\Rl)$. Then $\zeta\mapsto\langle f,\Delta^{-i\zeta/2\pi}\psi\rangle$ is analytic in the strip $\Strip_{\pi}$. Thus, for any closed path $\gamma\subset \Strip_{\pi}$, we get with the Theorems of Cauchy and Fubini
	\begin{align*}
		0
		&=
		\oint_\gamma d\zeta\,\langle f,\Delta^{-i\zeta/2\pi}\psi\rangle
		=
		\oint_\gamma d\zeta\int_\Rl d\te\, \overline{f(\te)}(\Delta^{-i\zeta/2\pi}\psi)(\te)
		=
		\int_\Rl d\te\, \overline{f(\te)}\oint_\gamma d\zeta(\Delta^{-i\zeta/2\pi}\psi)(\te)
		\,,
	\end{align*}
	which implies $\oint_\gamma d\zeta\,(\Delta^{-i\zeta/2\pi}\psi)(\te)=0$ for almost all $\te$ because $f$ was arbitrary. Thus we find $\te_0\in\Rl$ such that $\zeta\mapsto(\Delta^{-i\zeta/2\pi}\psi)(\te_0)=\psi(\te_0+\zeta)$ is analytic in $\Strip_{\pi}$, which implies the same analyticity for $\psi$. As $\psi\in\dom\Delta^{1/2}$, we see that $\psi$ is (the boundary value of) a function in ${\mathbb H}^2(\Strip_{\pi})$, and by $J\Delta^{1/2}\psi=\psi$, we also find $\overline{\psi(\te+i\pi)}=\psi(\te)$ for almost all $\te\in\Rl$.	
\end{proof}

\subsection{The standard subspace of an interval}\label{Section:Appendix1d}

Working in the lightray representation, we now turn to an explicit characterization of the interval subspaces $H(I)\subset L^2(\Rl_+,dp/p)$. As before, the skew-symmetric extension of a  function $\psi\in L^2(\Rl_+,dp/p)$ to $\Rl$ is denoted as
\begin{align}\label{eq:SkewExtension}
	\psi^{\tt s}(p) :=
	\begin{cases}
		\overline{\psi(-p)} & p<0\\
		\psi(p) & p>0
	\end{cases}
	\,.
\end{align}
Note that for $\psi\in L^2(\Rl_+,dp/p)$, one has $\psi^{\tt s}\in L^2(\Rl,dp/|p|)$, and $p\mapsto \psi^{\tt s}(p)/(1 + |p|)$ belongs to $L^2(\mathbb R , dp)$. So $\psi^{\tt s}$ is, in particular, a tempered distribution: $\psi^{\tt s}\in \Ss'(\Rl)$. In the following, we will often consider (inverse) Fourier transforms $\check{\psi}^{\tt s}$, as well as supports and derivatives of $\check{\psi}^{\tt s}$ in the sense of distributions.

Given an interval $I\subset \mathbb R$, we now consider the real linear subspace
\begin{align}
	K(I)
	:=
	\{\psi\in\Hil\,:\,\supp\check{\psi}^{\tt s}\subset \overline{I}\,\}
	\,,
\end{align}
where the inverse Fourier transform is taken as $\check{\psi}^{\tt s}(x)=(2\pi)^{-1/2}\int_\Rl dp\,\psi^{\tt s}(p)e^{-ipx}$, and correspondingly for distributions. In view of the skew symmetry of $\psi^{\tt s}$, its inverse Fourier transform is a real distribution.

\begin{lemma}\label{Lemma:KIClosed}
	$K(I)$ is a closed real linear subspace of $\Hil$.
\end{lemma}
\begin{proof}
	Let $\{\psi_n\}_{n\in\Nl}$ be a sequence in $K(I)$ that converges to $\psi\in L^2(\mathbb R_+ , dp/p)$. Thus
	\[
		\int_{-\infty}^{\infty} {\psi_n^{\tt s}}(p)\frac{\overline{h(p)}}{|p|}dp
		\longrightarrow
		\int_{-\infty}^\infty {\psi^{\tt s}}(p)\frac{\overline{h(p)}}{|p|}dp
	\]
	for every $h\in  L^2(\mathbb R , dp/|p|)$ as $n\to\infty$, and in particular for every $h\in \Ss(\mathbb R)$ such that $h(0)=0$. Thus
	\[
		\int_{-\infty}^{\infty} \psi_n^{\tt s}(p)\overline{k(p)}dp
		\longrightarrow
		\int_{-\infty}^\infty  \psi^{\tt s}(p)\overline{k(p)}dp
	\]
	for every $k\in\Ss(\mathbb R)$. It follows that $\psi_n^{\tt s}\to \psi^{\tt s}$ in $\Ss'(\mathbb R)$. Thus $\check{\psi}_n^{\tt s}\to \check{\psi}^{\tt s}$ in $\Ss'(\Rl)$, and hence the support of $\check{\psi}^{\tt s}$ is contained in $\overline{I}$, namely $\psi\in K(I)$.
\end{proof}

We now define, $a\in\Rl$,
\begin{align*}
	K(a,\infty):=\overline{\bigcup_{a<b}K(a,b)}^{\|\cdot\|_\Hil}
	\,,\qquad
	K(-\infty,a):=\overline{\bigcup_{b<a}K(b,a)}^{\|\cdot\|_\Hil}\,.
\end{align*}

\begin{lemma}\label{Lemma:H=KOnHalflines}
	$H(a,\infty)= K(a,\infty)$, and $H(-\infty, a) = K(-\infty, a)$, for every $a\in \mathbb R$.
\end{lemma}
\begin{proof}
	From the definition of the $K(I)$ it is clear that $I\mapsto K(I)$ is isotonous, and covariant under $U$. Thus $T(x)K(0,\infty)\subset K(0,\infty)$ for $x\geq0$ and $\Delta^{it}K(0,\infty)=K(0,\infty)$ for all $t\in\Rl$, and we can use the uniqueness statement of Proposition \ref{unique} to conclude $K(0,\infty) = \alpha H(0,\infty)$ for some $\alpha\in\Cl$ with $|\alpha|=1$. 
	
	Now $H(0,\infty)$ contains all functions of the form $\psi_f:=\widehat{f'}|_{\Rl_+}$, where $f\in C_{c,\Rl}^\infty(\Rl_+)$ and the dash denotes the derivative; this can be quickly checked on the basis of Lemma~\ref{Lemma:CharacterizeH-Appendix}. But then one sees that the inverse Fourier transform of $(\alpha\cdot\psi_f)^{\tt s}$ has support in $\Rl_+$ if and only if $\alpha$ is real, i.e. if $\alpha=\pm1$. Therefore $K(0,\infty) = \pm H(0,\infty)=H(0,\infty)$, and similarly  $H(-\infty, 0) = K(-\infty, 0)$. The statement for general $a$ follows by translation covariance.
\end{proof}

\noindent{\em Remark:} The argument in the above Lemma \ref{Lemma:H=KOnHalflines} also shows that
\begin{align}\label{densS}
	H(a,\infty) 
	&=
	\big\{ \widehat{f'}|_{\Rl_+}\,:\, f\in C_c^\infty(\mathbb R)  \text{ real and }\supp f \subset [a,\infty)\big\}^{\|\cdot\|_{\Hil}}\,,
	\\
	H(-\infty,a) 
	&=
	\big\{ \widehat{f'}|_{\Rl_+}\,:\, f\in C_c^\infty(\mathbb R)  \text{ real and }\supp f \subset (-\infty,a]\big\}^{\|\cdot\|_{\Hil}}\,.
	\label{eq:densS'}
\end{align}

\begin{corollary}\label{Hsupp}
	Let $\psi\in\Hil$. Then $\psi\in H$ if and only if $\supp\check\psi^{\tt s}\subset\Rl_+$.
\end{corollary}
\begin{proof}
	If $\psi\in H=K(0,\infty)$, there is a sequence $\{\psi_n\}_{n\in\Nl}\subset\bigcup_{a>0}K(0,a)$ such that $\|\psi_n-\psi\|\to 0$. Since the support of $\check\psi_n^{\tt s}$ is contained in $\Rl_+$, also the support of $\check\psi^{\tt s}$ is contained in $\Rl_+$ as in the proof of Lemma~\ref{Lemma:KIClosed}. Conversely, assume that $\check\psi^{\tt s}$ has  support  in $\mathbb R_+$. If $f\in C_c^\infty(\mathbb R)$ is real, with $\supp f\subset\Rl_-$,
	\[
		\int_{-\infty}^{\infty} \psi^{\tt s}(p)\overline{\widehat{f'}(p)}\frac{dp}{p}
		=
		i\int_{-\infty}^{\infty} \psi^{\tt s}(p)\overline{\widehat{f}(p)}\,dp
		=
		i\int_{-\infty}^{\infty}\check{\psi}^{\tt s}(x)f(x)dx = 0 \ .
	\]
	Since ${\psi}^{\tt s}$ and $\widehat{f'}$ are skew-symmetric, we then have
	\[ 
		\im \langle \widehat{f'},\psi \rangle
		=
		\im \int_{0}^{\infty} {\psi}(p)\overline{\widehat{f'}(p)}\frac{dp}{p}
		=\frac{1}{2i}\int_{-\infty}^{\infty} \psi^{\tt s}(p)\overline{\widehat{f'}(p)}\frac{dp}{p}
		=
		0\ .
	\]
	So $\psi\in H(-\infty, 0)' = H(0,\infty)$ by the density \eqref{eq:densS'}.
\end{proof}

By the definition of the subspaces $K(a,\infty)$, $K(-\infty,b)$, and by Lemma \ref{Lemma:H=KOnHalflines}, we see that  
\begin{align*}
	K(a,b)
	\subset
	K(a,\infty)\cap K(-\infty, b)
	=
	H(a,\infty)\cap H(-\infty, b)
	=
	H(a,b)
	,\qquad a<b\,,
\end{align*}
i.e. $I\mapsto K(I)$ is a subnet of $H\mapsto H(I)$. Indeed the two nets coincide:
\begin{lemma}\label{Lemma:K=H}
	$K(a,b) = H(a,b)$, $a<b$.
\end{lemma}
\begin{proof}
	Let $\psi\in \Hil$ belong to $H(a,b)\subset\Hil$. As $\psi\in H(a,\infty)=K(a,\infty)$,  the inverse Fourier transform $\check\psi^{\tt s}$ of $\psi^{\tt s}$ has  support in $[a,\infty)$. As $\psi$ also belongs to $H(-\infty,b)= K(-\infty,b)$ the support of $\check\psi^{\tt s}$ is contained in $(-\infty,b]$ too, thus $\check\psi^{\tt s}$ has support contained in $[a,b]$, namely $\psi\in K(a,b)$.
\end{proof}

Let $\psi\in K(I)$ for some interval $I$. Then, by the Paley-Wiener theorem, $\psi$ is the restriction of an entire analytic function to $\Rl_+$. Since $\check{\psi}$ is real, we also see that $\psi^{\tt s}$ coincides with the restriction of this entire function to $\Rl$. Furthermore, as $\psi\in  L^2(\mathbb R_+ , dp/p)$, we also have $\psi^{\tt s}(0)=0$. Consequently
\[
	p\mapsto i\frac{\psi^{\tt s}(p)}{p}\in L^2(\mathbb R, dp)\ .
\]
Now  the inverse Fourier transform of $p\mapsto i\psi^{\tt s}(p)/p$ has support in $\overline{I}$ too, because its derivative is $\check\psi^{\tt s}$ and $\psi\in L^2(\Rl_+,dp/p)$. Thus any $\psi\in H(I)$ can be represented in the form $\psi(p)=-ip\fhat(p)$ for some real $f\in L^2(\overline{I},dx)$.

By the Paley-Wiener theorem for $L^2$ \cite[Thm.~4.1]{SteinWeiss:1971} we then have with a suitable constant $C>0$ the bound
\[
	|\psi(p)| \leq C |p| e^{r\,|\im p|}\ , \quad p\in\mathbb C\ ,
\]
if $I = I_r = [-r, r]$. Summarizing, we have the following characterization of $H(I_r)$.

\begin{proposition}\label{Proposition:CharacterizationOfHI-Appendix}
	Let $(H,T)$ be the one-dimensional irreducible standard pair, and $\psi\in\Hil = L^2(\mathbb R_+ , dp/p)$. The following are equivalent:
	\begin{enumerate}
		\item $\psi\in H(I_r)$.
		\item The inverse Fourier transform of $\psi^{\tt s}$ has support in $I_r$.
		\item $\psi$ extends to an entire analytic function such that $\psi(-\bar p) = \overline{\psi(p)}$ and
		$|\psi(p)| \leq C |p| e^{r|\im p|}$ for all $p\in\mathbb C$ and some constant $C>0$.
% 		\item $\psi$ extends to an entire analytic function such that $\psi(-\bar p) = \overline{\psi(p)}$ and for some $N\in\Nl$
% 		$|\psi(p)| \leq C (1 + |p|)^N e^{r|\im p|}$, $p\in\mathbb C$, for some constant $C>0$.
		\item $\psi$ extends to an entire analytic function such that $\psi(-\bar p) = \overline{\psi(p)}$, and of exponential type at most $r$, namely for all $\eps>0$ there exists $C_\eps>0$ such that $|\psi(p)| \leq C_\varepsilon  e^{(r+\varepsilon)|p|}$,  $p\in\mathbb C$.
	\end{enumerate}
	For $\psi\in H(I)$, the function $\Rl_+\ni p\mapsto -i\psi(p)/p$ is the restriction to $\Rl_+$ of the Fourier transform of a real function in $L^2(\Rl,dx)$ with support in the closure of $I$, and
	\begin{align}\label{eq:HIFourier1d}
		\{\widehat{f'}|_{\Rl_+}\,:\,f\in C_{c,\Rl}^\infty(I)\}
		\subset 
		H(I)
	\end{align}
	is cyclic.
\end{proposition}
\begin{proof}
	We have seen that $i)\Leftrightarrow ii)\Rightarrow iii)$, and $iii)\Rightarrow iv)$ is obvious. If $iv)$ holds, also $\xi(p):=i\psi^{\tt s}(p)/p$ is an entire function of exponential type at most $r$ because $\psi^{\tt s}(0)=0$. Thus, by the Paley-Wiener theorem for $L^2$ \cite[Thm.~4.1]{SteinWeiss:1971}, $\supp\check{\xi}\subset I_r$ and hence also its derivative $\check{\psi}^{\tt s}$ is supported in $I_r$, i.e. $ii)$ holds.
		
	The statement about the representation of functions in $H(I)$ by Fourier transforms of $L^2(\overline{I},dx)$ has been shown before, and implies \eqref{eq:HIFourier1d} by the density of $C_c^\infty(I)\subset L^2(\overline{I},dx)$.
\end{proof}

\subsection{The standard subspace of a double cone}\label{Section:Appendix2d}

We next consider the structure of standard subspaces for the two-dimensional standard pair with associated massive irreducible representation $U_m$. We will mainly be working in the momentum picture. The characterization of the double cone spaces $H^m(\OO_r)$ is partly similar to the characterization of the interval spaces $H(I_r)$ obtained in Proposition~\ref{Proposition:CharacterizationOfHI-Appendix}; we will be brief about the analogous parts of the proofs.

As in the preceding section, any $\psi\in\Hil=L^2(\Rl,\om_m(p_1)^{-1}dp_1)$ can be considered as a tempered distribution. Making use of the decomposition of vectors $\psi\in\Hil$ into the $\Gamma$-invariant components $\psi_\pm$ \eqref{eq:psi+-MT}, we define
\begin{align}
	K(\OO_{a,b})
	:=
	\{\psi\in\Hil\,:\,\supp\check\psi_\pm\subset [a,b]\}\,,
\end{align}
where $\OO_{a,b}\subset\Rl^2$ is the double cone with base $[a,b]$ on the time zero line. Taking into account the continuity of $\psi\mapsto\psi_\pm$, one proves as in Lemma~\ref{Lemma:KIClosed} that $K(\OO_{a,b})$ is a {\em closed} real linear subspace of $\Hil$.

We also introduce the real subspaces
\begin{align}
	L(\OO)
	:=
	\{p_1\mapsto \fti(\om_m(p_1),p_1)\,:\,f\in C_{c,\Rl}^\infty(\OO)\}\,,
\end{align}
where $\OO\subset\Rl^2$ is any double cone or wedge. It is clear from this definition and the form of $U_m$ that $\OO\mapsto L(\OO)$ satisfies isotony and transforms covariantly under $U_m$. We also recall that $L(\OO_{a,b})\subset K(\OO_{a,b})$ is a dense subspace of $K(\OO_{a,b})$, as can be shown by investigating the Cauchy problem for the Klein-Gordon equation.

We then define subspaces associated with the spatially translated wedges $\pm W_2+c:=\pm W_2+({0\atop c})$, $c\in\Rl$, as
\begin{align}
	K(W_2+c):=\overline{\bigcup_{c<a<b}K(\OO_{a,b})}^{\|\cdot\|}
	\,,\qquad
	K(-W_2+c):=\overline{\bigcup_{a<b<c}K(\OO_{a,b})}^{\|\cdot\|}
	\,.
\end{align}
Similarly to Lemma~\ref{Lemma:H=KOnHalflines} and Corollary~\ref{Hsupp}, we find
\begin{lemma}
	\begin{enumerate}
		\item $H(\pm W_2+c)=K(\pm W_2+c)$ for all $c\in\Rl$.
		\item Let $\psi\in\Hil$. Then $\psi\in H$ if and only if $\supp\check\psi_\pm\subset\Rl_+$.
	\end{enumerate}
\end{lemma}
\begin{proof}
	{\em i)} As the spatial translations commute with $\Gamma$, it is clear that the net $K$ is covariant under spatial translations, and it is therefore sufficient to consider the case $c=0$.
	
	We first argue that $K(W_2)$ is invariant under the modular unitaries $\Delta^{it}$, $t\in\Rl$. In fact, we have $\Delta^{it}L(\OO)=L(\la_2(t)\OO)$ by the covariance of the $L$-net, and for $\OO=\OO_{a,b}$ with $0<a<b$, we find by isotony $0<a(t)<b(t)$ such that $\la_2(t)\OO_{a,b}\subset\OO_{a(t),b(t)}$. As $L(\OO_{a,b})\subset K(\OO_{a,b})$ is dense and $\Delta^{it}$ is unitary, this yields $\Delta^{it}K(\OO_{a,b})\subset K(\OO_{a(t),b(t)})$ and thus $\Delta^{it}K(W_2)=K(W_2)$, $t\in\Rl$.
	
	Next, we see from the spatial translational invariance of the $K$-net that $T_m(0,x_1)K(W_2)\subset K(W_2)$ for $x_1>0$. Since any $x\in W_2$ is of the form $x=\la_2(t)({0\atop a})$ for suitable $t\in\Rl$ and $a>0$, this implies via the commutation relations \eqref{eq:BorchersCommutationRelations} that $T_m(x)K(W_2)\subset K(W_2)$ for all $x\in W_2$.
	
	We are thus in the position to apply Proposition~\ref{unique} to conclude $K(W_2)=\alpha H(W_2)=\alpha H$ for some $\alpha\in\Cl$ with $|\alpha|=1$. Taking then $f\in C_{c,\Rl}^\infty(W_2)$, the function $p_1\mapsto\fti(\om_m(p_1),p_1)$ lies in $K(W_2)$ and in $H$ (the latter claim can be quickly checked on the basis of Lemma~\ref{Lemma:CharacterizeH-Appendix}~{\em i)} by changing coordinates according to $p_1=m\sinh\te$), and $\alpha \psi\in H$ if and only if $\alpha$ is real. Thus $\alpha=\pm1$ and $K(W_2)=H$. The proof for the left wedge is analogous.
	
	{\em ii)} As in Corollary~\ref{Hsupp}, it is clear that $\psi\in H$ implies $\supp\check\psi_\pm\subset\Rl_+$. For the converse direction, we first note that for any $\psi,\xi\in\Hil$, the scalar products $\langle\psi_+,\xi_+\rangle$ and $\langle i\om\psi_-,i\om\xi_-\rangle$ are real, because both their entries are eigenvectors of the anti unitary involution $\Gamma$ with the same eigenvalue $\pm1$, and the scalar products $\langle\psi_+,i\om\xi_-\rangle$, $\langle i\om\psi_-,\xi_+\rangle$ are imaginary because their entries are eigenvalues to different eigenvalues of $\Gamma$. Thus
	\begin{align*}
		{\rm Im}\langle\psi,\xi\rangle
		&=
		{\rm Im}\langle\psi_++i\om\psi_-,\xi_++i\om\xi_-\rangle
		\\
		&=
		\langle\psi_+,\om\xi_-\rangle-\langle\om\psi_-,\xi_+\rangle
		\\
		&=
		\int_\Rl dx_1\left({\check{\psi}_+(x_1)}\check{\xi}_-(x_1)-{\check{\psi}_-(x_1)}\check{\xi}_+(x_1)\right)
		\,.
	\end{align*}
	If we now choose $\psi\in\Hil$ with $\supp\check{\psi}_\pm\subset\Rl_+$, and $\xi\in L(-W_2)$, i.e. $\check{\xi}_\pm\in C_{c,\Rl}^\infty(\Rl_-)$, then  Im$\langle\psi,\xi\rangle=0$. As $L(-W_2)\subset K(-W_2)=H'$ is dense, this implies $\psi\in H''=H$.
\end{proof}

With this characterization of the wedge subspaces, we can now argue in complete analogy to Lemma~\ref{Lemma:K=H} that
\begin{align}
	K(\OO_{a,b})=T_m(0,a)H\cap T_m(0,b)H'=H^m(\OO_{a,b})\,,\qquad a<b\,.
\end{align}

We conclude by summarizing our characterization of the double cone subspaces for symmetric double cones $O_r=\OO_{-r,r}$ with $a=-r$, $b=r$.

\begin{proposition}\label{Proposition:CharacterizationOfHO-Appendix}
	Let $(H,T_m)$ be the two-dimensional irreducible standard pair with representation $U_m$, $m>0$, and $\psi\in\Hil = L^2(\Rl,\om_m(p_1)^{-1}dp_1)$. The following are equivalent:
	\begin{enumerate}
		\item $\psi\in H^m(\OO_r)$.
		\item The inverse Fourier transforms of $\psi_\pm$ \eqref{eq:psi+-MT} have support in $I_r$.
		\item The functions $\psi_\pm$ \eqref{eq:psi+-MT} extend to entire analytic functions such that $\psi_\pm(-\overline{p_1}) = \overline{\psi_\pm(p_1)}$ and $|\psi_\pm(p_1)| \leq C(1+|p_1|)^N  e^{r|{\rm Im}\,p_1|}$,  $p_1\in\mathbb C$.
		\item The functions $\psi_\pm$ \eqref{eq:psi+-MT} extend to entire analytic functions such that $\psi_\pm(-\overline{p_1}) = \overline{\psi_\pm(p_1)}$ and of exponential type at most $r$, namely for all $\eps>0$ there exists $C_\eps>0$ such that $|\psi_\pm(p_1)| \leq C_\varepsilon  e^{(r+\varepsilon)|p_1|}$,  $p_1\in\mathbb C$.
	\end{enumerate}
	The subspace
	\begin{align}\label{eq:HIFourier}
		\{\psi\,:\,\check{\psi}_\pm\in C_{c,\Rl}^\infty(I_r)\}
		\subset 
		H^m(\OO_r)
	\end{align}
	is cyclic.
\end{proposition}
\begin{proof}
	We have seen $i)\Leftrightarrow ii)$ already. If $ii)$ holds, the functions $\psi_\pm$ satisfy the reality condition $\psi_\pm(-\overline{p_1}) = \overline{\psi_\pm(p_1)}$ because $\check{\psi}_\pm$ are real distributions, and $\psi_\pm$ extend to entire analytic functions satisfying $|\psi_\pm(p)|\leq C(1+|p_1|)^N\,e^{r|{\rm Im}\,p_1|}$, $p_1\in\Cl$, by the Paley-Wiener Theorem for distributions, see \cite[Thm.~IX.12]{ReedSimon:1975}, i.e. we have shown $ii)\Rightarrow iii)$, and $iii)\Rightarrow iv)$ is trivial.
	
	It remains to show $iv)\Rightarrow ii)$. To this end, we first note that $f_\pm(p_1):=\psi_\pm(p_1)/(im+p_1)$ is analytic and of exponential type at most $r$ on the upper half plane, and $\int_\Rl dp_1\,|f_\pm(p_1)|^2\leq\|\psi_\pm\|^2/m<\infty$. This implies \cite[Thm.~6.7.7]{Boas:1954} $\int_\Rl dp_1\,|f_\pm(p_1+iq)|^2\leq e^{2rq}\|\psi_\pm\|^2/m$ for $q\geq0$, from which we read off that $p_1\mapsto e^{iRp_1}f_\pm(p_1)$ lies in the Hardy space ${\mathbb H}^2(\Cl_+)=\widehat{L^2(\Rl_+)}$ of the upper half plane if $R>r$. Hence $\supp\check{f}_\pm\subset[-R,\infty)$ for all $R>r$, i.e. $\supp\check{f}\subset[-r,\infty)$. Thus also $\check{\psi}_\pm=im\check{f}_\pm+i\check{f}_\pm'$ has support in $[-r,\infty)$. Arguing analogously with $g_\pm(p_1):=\psi_\pm(-p_1)/(im+p_1)$ yields also $\supp\check{\psi}_\pm\subset(-\infty,r]$, i.e. we have $\supp\check{\psi}_\pm\subset I_r$.
\end{proof}

It is also useful to translate this characterization to the lightray picture.

\begin{corollary}\label{Corollary:CharacterizationOfHO-lightray-appendix}
	In the lightray representation, let $\psi\in\Hil= L^2(\Rl_+,dp/p)$. The following are equivalent:
	\begin{enumerate}
		\item $\psi\in H^m(\OO_r)$
		\item There exist $\psi^\pm\in\Hil$, $\overline{\psi^\pm(\frac{1}{p})}=\overline{\psi_\pm(-p)}=\psi^\pm(p)$, $p>0$, which are restrictions to $\Rl_+$ of functions analytic on $\Cl\backslash\{0\}$ and satisfy for any $\eps>0$ the bound
		\begin{align}
			|\psi^\pm(p)|\leq C_{\eps,\pm}\,e^{(\frac{mr}{2}+\eps)|p-\frac{1}{p}|}
			\,,\qquad
			p\in\Cl\backslash\{0\}\,,
		\end{align}
		for some $C_{\eps,\pm}>0$, such that
		\begin{align}\label{eq:CauchyRepPsi-LightRay}
			\psi(p)
			=
			\psi^+(p)+i\,(p+\tfrac{1}{p})\cdot\psi^-(p)
			\,,\qquad
			p>0\,.
		\end{align}
	\end{enumerate}
\end{corollary}
\begin{proof}
	$i)\Rightarrow ii)$. This follows by translating the characterization of $H^m(\OO_r)$ in Proposition~\ref{Proposition:CharacterizationOfHO-Appendix}~{\em iv)} to the lightray representation, with $p_1=\frac{m}{2}(p-\frac{1}{p})$: We have 
	\begin{align*}
		\psi(p)
		=
		\psi_+(\tfrac{m}{2}(p-\tfrac{1}{p}))+i\tfrac{m}{2}(p+\tfrac{1}{p})\cdot \psi_-(\tfrac{m}{2}(p-\tfrac{1}{p}))
		\,,
	\end{align*}
	with $\psi^\pm:p\mapsto\psi_\pm(\tfrac{m}{2}(p-\tfrac{1}{p}))\in\Hil$ satisfying $\overline{\psi^\pm(\frac{1}{p})}=\overline{\psi^\pm(-p)}=\psi^\pm(p)$, $p>0$, because of $\overline{\psi_\pm(-p_1)}=\psi_\pm(p_1)$, $p_1\in\Rl$. Furthermore, $\psi^\pm$ have the claimed analyticity and boundedness properties thanks to Proposition~\ref{Proposition:CharacterizationOfHO-Appendix}~{\em iv)}.
	
	For the converse direction, we start from $\psi^\pm\in\Hil$ with the specified properties, and have to show that $\psi_\pm(p_1):=\psi^\pm(\frac{1}{m}(\om_m(p_1)-p_1))$ satisfy the conditions of Proposition~\ref{Proposition:CharacterizationOfHO-Appendix}~{\em iv)}. It is clear that $\psi_\pm$ lie in $L^2(\Rl,\om_m(p_1)^{-1}dp_1)$, and are invariant under $\Gamma$. From their definition, we also see analyticity of $\psi_\pm$ on $\Cl\backslash\{\pm i[m,\infty)\}$ because of the branch cut of the root in $\om_m(p_1)=(m^2+p_1^2)^{1/2}$. The boundaries of these cuts correspond to imaginary values of $p$, and in view of $\overline{\psi^\pm(\frac{1}{\overline{p}})}=\overline{\psi^\pm(-\overline{p})}=\psi^\pm(p)$, the boundary values on the two opposite sides of the (purely imaginary) cuts are real and identical, and $\psi_\pm$ is skew symmetric w.r.t. reflecting about the cut. Thus we can use Schwarz' reflection principle to conclude that $\psi_\pm$	extends to an entire function. Since $|p-\frac{1}{p}|=2|p_1|$, it 
is also easy to see that $\psi_\pm$ is of exponential type at most $r$. By Proposition~\ref{Proposition:CharacterizationOfHO-Appendix}~{\em iv)}, it then follows that $\psi\in H^m(\OO_r)$.
\end{proof}

\footnotesize
%%================================================
% \bibliography{Bibtex-Database.bib}
% \bibliographystyle{alphaGL}

%%================================================

\end{document}